\documentclass[runningheads]{llncs}

\usepackage{graphicx}

\usepackage{hyperref}

\usepackage{amsthm} %
\usepackage{cite}
\usepackage{algorithm}
\usepackage{algorithmic}
\usepackage{textcomp}
\usepackage{xcolor}
\usepackage{tabularx}

\usepackage{booktabs} %
\usepackage[utf8]{inputenc}
\usepackage{tikz}
\usepackage{enumitem}
\usepackage{amssymb}
\usepackage{amsmath}
\usepackage{pgf}
\usepackage{caption}
\usepackage{subfig}
\usepackage{thm-restate}
\usepackage{cleveref}
\usepackage{flushend}
\usepackage{balance}
\usepackage{ifthen}

\usetikzlibrary{calc}
\usetikzlibrary{decorations.pathreplacing}
\usetikzlibrary{shapes,arrows}
\usetikzlibrary{shadows}
\usetikzlibrary{decorations.pathmorphing} %
\usetikzlibrary{fit}					%
\usetikzlibrary{backgrounds}	%

\newboolean{final}
\setboolean{final}{false}

\usepackage{thmtools}

\declaretheorem[name=Definition]{Definition}
\numberwithin{Definition}{section}
\declaretheorem[name=Property,sibling=Definition]{Property}
\numberwithin{Property}{section}

\numberwithin{Hypothesis}{section}

\numberwithin{Theorem}{section}
\declaretheorem[name=Corollary,sibling=Definition]{Corollary}
\numberwithin{Corollary}{section}
\declaretheorem[name=Example,sibling=Definition]{Example}
\numberwithin{Example}{section}

\numberwithin{Proposition}{section}
\declaretheorem[name=Lemma,sibling=Definition]{Lemma}
\numberwithin{Lemma}{section}

\numberwithin{Grammar}{section}

\numberwithin{Algorithm}{section}

\numberwithin{Assumption}{section}

\let\oldnl\nl%
\newcommand{\noln}{\renewcommand{\nl}{\let\nl\oldnl}}

\newcommand{\ignore}[1]{}

\renewcommand{\paragraph}[1]{\noindent\textbf{#1.}}

\definecolor{ao(english)}{rgb}{0.0, 0.5, 0.0}

\numberwithin{equation}{section}

\begin{document}
\ifthenelse{\boolean{final}}{
    \title{Equivalent Rewritings\\ on Path Views with Binding Patterns}}
    {\title{Equivalent Rewritings on Path Views\\ with Binding Patterns (Extended Version)}}

\author{Julien Romero\inst{1}
\and
Nicoleta Preda\inst{2} \and
Antoine Amarilli\inst{1} \and
Fabian Suchanek\inst{1}
}

\authorrunning{J. Romero et al.}
\institute{LTCI, Télécom Paris, Institut Polytechnique de Paris, France \email{first.last@telecom-paris.fr} \and
Université de Versailles  \email{nicoleta.preda@uvsq.fr}}

\maketitle              %
\begin{abstract}
	A view with a binding pattern is a parameterized query on a database. Such views are used, e.g., to model Web services. To answer a query on such views, the views have to be orchestrated together in execution plans. We show how queries can be rewritten into equivalent execution plans, which are guaranteed to deliver the same results as the query on all databases. We provide a correct and complete algorithm to find these plans for path views and atomic queries. Finally, we show that our method can be used to answer queries on real-world Web services.
\end{abstract}

\begin{figure*}
	\centering
	\tikzstyle arrowstyle=[blue,semitransparent,scale=2]
\tikzstyle basiclabel=[draw=none,fill=none,shape=rectangle,inner sep=2pt,scale=.8]
\tikzstyle leftlabel=[basiclabel,anchor=east]
\tikzstyle rightlabel=[basiclabel,anchor=west]
\tikzstyle bottomlabel=[basiclabel,anchor=north]
\tikzstyle toplabel=[basiclabel,anchor=south]

\definecolor{lightblue}{cmyk}{0.12,0,0,0.15} 
\tikzstyle{block} = [rectangle, draw, fill=white, 
text centered, rounded corners]

\definecolor{darkpastelgreen}{rgb}{0.01, 0.75, 0.24}
\definecolor{dartmouthgreen}{rgb}{0.05, 0.5, 0.06}
\definecolor{forestgreen}{rgb}{0.0, 0.27, 0.13}
\definecolor{lasallegreen}{rgb}{0.03, 0.47, 0.19}

\tikzstyle{headvar} = [rectangle, draw, fill=lightblue!70,
text centered, rounded corners]    
\tikzstyle{labelcase} = [rectangle,    text centered]
\tikzstyle{background}=[rectangle,   fill=lightblue!70,
inner sep=0.2cm,
rounded corners=5mm]

\tikzstyle{line} = [draw, -latex']
\tikzstyle{linenoarrow} = [draw]
\tikzstyle{invisibleline} = [-latex',sloped]
\tikzstyle{dashedline} = [draw, dashed]
\tikzstyle{inputtar} = [rectangle, draw, fill=lightblue!70, 
text centered, rounded corners]

\tikzstyle{inputvar} = [rectangle, draw=red,  fill=red!30,
text centered, rounded corners]

\begin{tikzpicture}[scale=0.9, transform shape]

\node [inputvar] (z) {\textit{Jailhouse}};
\node [block] (y) at ($(z)+(+3.7cm,0cm)$) {\textit{Jailhouse Rock}};
\node [block] (x) at ($(z)+(-3cm,0cm)$) {\textit{Elvis Presley}};
\node [block] (t) at ($(y)+(-0.5cm,-2.5cm)$){\textit{I Walk the Line}};
\path [line] (z) --node [above,align=center] {{\textit{onAlbum}} } (y);
\path [line] (x) --node [above,align=center] {{\textit{sang}} } (z);
\path [line] (z) --node [sloped,above, align=center] {{\textit{relatedAlbum}} } (t);

\coordinate (zu) at   ($(z) + (+0.6cm,+0.4cm) $); 
\coordinate (yu) at   ($(y) + (+0cm,+0.4cm) $); 
\draw [blue,->]     (zu) to  (yu);
\node [labelcase,blue] (lf1) at   ($0.5*(zu) + 0.5*(yu) + (0,+0.2cm)$) {\textit{getAlbum}};

\coordinate (xb) at   ($(x) + (+0cm,+0.8cm) $);
\coordinate (yb) at   ($(y) +  (+0cm,+0.8cm) $);
\draw [->,blue]    (yb) to (xb);
\node [labelcase,blue] (lf2) at   ($0.5*(xb) + 0.5*(yb) + (0cm,+0.19cm)$) {\textit{getAlbumDetails}};

\node [block] (z2) at ($(t)+(-5cm,0cm)$) {\textit{Folsom Prison Blues}};
\node [block] (x2) at ($(z2)+(-4cm,0cm)$) {\textit{Johnny Cash}};

\path [line] (z2) --node [above,align=center] {{\textit{onAlbum}} } (t);
\path [line] (x2) --node [above,align=center] {{\textit{sang}} } (z2);

\coordinate (xb2) at   ($(x2) + (+0cm,+0.4cm) $);
\coordinate (yb2) at   ($(t) +  (-0.9cm,+0.4cm) $);
\draw [->,lasallegreen]    (yb2) to (xb2);
\node [labelcase,lasallegreen] (lf22) at   ($0.5*(xb2) + 0.5*(yb2) + (0cm,+0.2cm)$) {\textit{getAlbumDetails}};

\coordinate (zr) at   ($(z) + (-0cm,-0.3cm) $); 
\coordinate (tr) at   ($(yb2) + (0cm,0cm) $); 
\path [line, lasallegreen] (zr) --node [sloped,below, align=center] {{\textit{getRelAlbum}} } (tr);

\end{tikzpicture} 
	\caption{An equivalent execution plan (blue) and a maximal contained rewriting (green) executed on a database (black). \label{fig:ex1}}\vspace{-5mm}
\end{figure*}
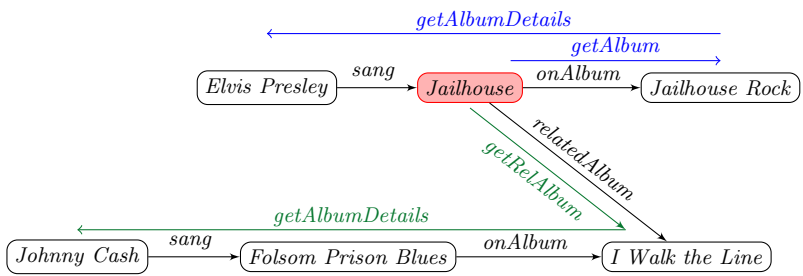

\section{Introduction}
In this paper, we study views with binding patterns~\cite{BINDING-PATTERNS-PODS1995}.
Intuitively, these can be seen as functions that, given input values, return output values from a database. For example, a function on a music database could take as input a musician, and return the songs by the musician stored in the database. 

Several databases on the Web can be accessed only through such functions. They are usually presented as a form or as a Web service. For a REST Web service, a client calls a function by accessing a parameterized URL, and it responds by sending back the results in an XML or JSON file. The advantage of such an interface is that it offers a simple way of accessing the data without downloading it. Furthermore, the functions allow the data provider to choose which data to expose, and under which conditions. For example, the data provider can allow only queries about a given entity, or limit the number of calls per minute. According to \href{http://programmableweb.com}{programmableWeb.com}, there are over 20,000 Web services of this form -- including LibraryThing, Amazon, TMDb, Musicbrainz, and Lastfm.

If we want to answer a user query on a database with such functions, we have to \emph{compose} them. For example, consider a database about music -- as shown in Figure~\ref{fig:ex1} in black. Assume that the user wants to find the musician of the song \emph{Jailhouse}. One way to answer this query is to call a function \emph{getAlbum}, which returns the album of the song. Then we can call \emph{getAlbumDetails}, which takes as input the album, and returns all songs on the album and their musicians. If we consider among these results only those with the song \emph{Jailhouse}, we obtain the musician \emph{Elvis Presley} (Figure~\ref{fig:ex1}, top, in blue). We will later see that, under certain conditions, this plan is guaranteed to return exactly all answers to the query on all databases: it is an \emph{equivalent rewriting} of the query.
This plan is in contrast to other possible plans, such as calling \emph{getRelatedAlbum} and \emph{getAlbumDetails} (Figure~\ref{fig:ex1}, bottom, in green). This plan does not return the exact set of query results. It is a \emph{maximally contained rewriting}, another form of rewriting, which we will discuss in the related work.

Equivalent rewritings are of primordial interest to the user because they allow obtaining exactly the answers to the query -- no matter what the database contains.
Equivalent rewritings are also of interest to the data provider: For example, in the interest of usability, the provider may want to make sure that equivalent plans can answer all queries of importance. 
However, finding equivalent rewritings is inherently non-trivial.  As observed in \cite{benedikt2015querying,benedikt2016generating}, the problem is undecidable in general. Indeed, plans can recursively call the same function. Thus, there is, a priori, no bound on the length of an execution plan. Hence, if there is no plan, an algorithm may try forever to find one -- which indeed happens in practice.%

In this paper, we focus on path functions (i.e., functions that form a sequence of relations) and atomic queries. For this scenario, we can give a correct and complete algorithm that decides in PTIME whether a query has an equivalent rewriting or not. If it has one, we can give a grammar that enumerates all of them. Finally, we show that our method can be used to answer queries on real-world Web services. 
After reviewing related work in Section~\ref{sec-rel} and preliminaries in Section~\ref{sec-prel}, we present our problem statement in Section~\ref{sec:problem} and our algorithm in Section~\ref{sec:thealgorithm}, concluding with experiments in Section~\ref{sec-experiments}.
\ifthenelse{\boolean{final}}{This paper is complemented by an extended version~\cite{romero2020equivalent} that contains the proofs for our theorems.}{This is an extended version of the conference paper which contains all detailed proofs in the appendix.}
\section{Related Work}\label{sec-rel}

Formally, we aim at computing \emph{equivalent rewritings} over views with binding patterns~\cite{BINDING-PATTERNS-PODS1995} in the presence of inclusion dependencies. Our approach relates to the following other works. %

\paragraph{Equivalent Rewritings} Checking if a query is \emph{determined} by views~\cite{gogacz2016red}, or finding possible equivalent rewritings of a query over views, is a task that has been intensively studied for query optimization~\cite{Dana-Halevy-OptBindingPatterns-99,benedikt2015querying}, under various classes of constraints. In our work, we are specifically interested in computing equivalent rewritings over views with binding patterns, i.e., restrictions on how the views can be accessed. This question has also been studied, in particular with the approach by Benedikt et al.~\cite{benedikt2016generating} based on logical interpolation, for very general classes of constraints. In our setting, we focus on path views and unary inclusion dependencies on binary relations.
This restricted (but practically relevant) language of functions and constraints has not been investigated in~\cite{benedikt2016generating}. We show 
that, in this context, the problem is solvable in PTIME. What is more, we provide a self-contained, effective algorithm for computing plans, for which we provide an implementation. We compare experimentally against the PDQ implementation by Benedikt et al.~\cite{benedikt2014pdq}\ in Section~\ref{sec-experiments}.

\paragraph{Maximally Contained Rewritings} Another line of work has studied how to rewrite queries against data sources in a way that is not equivalent but maximizes the number of query answers~\cite{HalevyAnswering2001}. Unlike equivalent rewritings, there is no guarantee that all answers are returned. 
For views with binding patterns, a first solution was proposed in~\cite{DuschkaL97,DuschkaGL00}.
The problem has also been studied for different query languages or under various constraints~\cite{CaliICDE2008, NashL04,CaliCM09,DeutschLN07}. We remark that by definition, the approach requires the generation of relevant but not-so-smart call compositions. These call compositions make sure that no answers are lost. 
Earlier work by some of the present authors proposed to prioritize promising function calls \cite{angie} or to complete the set of functions with new functions~\cite{susie}. In our case, however, we are concerned with identifying only those function compositions that are guaranteed to deliver answers.

\paragraph{Orthogonal Works} Several works study how to optimize given execution plans~\cite{Srivastava06,WEB_COMP_DATA_INTEGRATION_VLDB2005}. Our work, in contrast, aims at \emph{finding} such execution plans. Other works are concerned with mapping several functions onto the same schema~\cite{ontology-survey,TaheriyanKSA12,KoutrakiVP15}. Our approach takes a Local As View perspective, in which all functions are already formulated in the same schema.

\paragraph{Federated Databases}  Some works~\cite{QuilitzL08,SchwarteHHSS11} have studied \emph{federated databases}, where each source can be queried with any query from a predefined language. By contrast, our sources only publish a set of preset parameterized queries, and the abstraction for a Web service is a view with a binding pattern, hence, a predefined query with input parameters. Therefore, our setting is different from theirs, as we cannot send arbitrary queries to the data sources: we can only call these predefined functions.

\paragraph{Web Services} There are different types of Web services, and many of them are not (or cannot be) modeled as views with binding patterns. AJAX Web services use JavaScript to allow a Web page to contact the server. Other Web services are used to execute complex business processes~\cite{2012Deutsch} according to protocols or choreographies, often described in BPEL~\cite{bpel}.
The Web Services Description Language (WSDL) describes SOAP Web services.
The Web Services Modeling Ontology (WSMO)~\cite{WSDMO}, in the Web Ontology Language for Services (OWL-S)~\cite{MartinPMBMMPPSSSS04}, or in Description Logics (DL)~\cite{RaoKM04} can describe more complex services. These descriptions allow for Artificial Intelligence reasoning about Web services in terms of their behavior by explicitly declaring their preconditions and effects. Some works derive or enrich such descriptions automatically~\cite{PuHK06,DBLP:conf/icwe/CeriBB11,DBLP:conf/www/BozzonBC12} in order to facilitate Web service discovery.

In our work, we only study Web services that are querying interfaces to databases. These can be modeled as views with binding patterns and are typically implemented in the Representational State Transfer (REST) architecture, which does not provide a formal or semantic description of the functions.

\section{Preliminaries} \label{sec-prel}

\paragraph{Global Schema} 
We assume a set $\mathcal{C}$ of constants and a set $\mathcal{R}$ of relation names. 
We assume that all relations are binary, i.e., any $n$-ary relations have been encoded as binary relations by introducing additional constants\footnote{\url{https://www.w3.org/TR/swbp-n-aryRelations/}}.
A \textit{fact} $r(a, b)$ is formed using a relation name $r \in \mathcal{R}$ and two constants $a, b \in \mathcal{C}$. A
\textit{database instance} $I$, or simply \textit{instance}, is a set of facts. For $r\in \mathcal{R}$, we will use $r^-$ as a relation name to mean the inverse of $r$, i.e., $r^-(b, a)$ stands for $r(a, b)$. More precisely, we see the inverse relations $r^-$ for $r\in\mathcal{R}$ as being relation names in $\mathcal{R}$, and we assume that, for any instance $I$, the facts of $I$ involving the relation name $r^-$ are always precisely the facts $r^-(b,a)$ such that $r(a,b)$ is in $I$.

\paragraph{Inclusion Dependencies} %
A \emph{unary inclusion dependency} for two relations $r, s$, which we write $r\leadsto s$,
is the following constraint:
\[ \forall x, y: r(x, y) \Rightarrow \exists z: s(x, z)\]
Note that one of the two relations or both may be inverses.
In the following, we will assume a fixed set $\mathcal{UID}$ of unary inclusion dependencies, and we will only consider instances that satisfy these inclusion dependencies. We assume that $\mathcal{UID}$ is closed under implication, i.e., if $r \leadsto s$ and $s \leadsto t$ are two inclusion dependencies in $\mathcal{UID}$, then so is $r \leadsto t$. 

\paragraph{Queries} An \textit{atom} $r(\alpha,\beta)$ is formed with a relation name $r \in \mathcal{R}$ and $\alpha$ and $\beta$ being either constants or variables. A \emph{query} takes the form
\[q(\alpha_1,...,\alpha_m) \leftarrow  B_1,...,B_n\]
where $\alpha_1,...\alpha_m$ are variables, each of which must appear in at least one of the body atoms $B_1,...B_n$. We assume that queries are \emph{connected}, i.e., each body atom must be transitively linked to every other body atom by shared variables.
An \emph{embedding} for a query $q$ on a database instance $I$ is a substitution $\sigma$ for the variables of the body atoms so that $\forall B\in\{B_1,...,B_n\}: \sigma(B)\in I$. A \emph{result} of a query is an embedding projected to the variables of the head atom. We write $q(\alpha_1,...,\alpha_m)(I)$ for the results of the query on $I$. An \emph{atomic query} is a query that takes the form $q(x) \leftarrow r(a, x)$, where $a$ is a constant and $x$ is a variable.

\paragraph{Functions} We model functions as views with binding patterns~\cite{BINDING-PATTERNS-PODS1995}, namely:
\[f(\underline{x}, y_1,...,y_m) \leftarrow  B_1,...,B_n\]
Here, $f$ is the function name, $x$ is the \emph{input variable} (which we underline), $y_1,...,y_m$ are the \emph{output variables}, and any other variables of the body atoms are \emph{existential variables}. 
In this paper, we are concerned with \emph{path functions}, where the body atoms are ordered in a sequence $r_1(\underline{x}, x_1), r_2(x_1,x_2),...,r_n(x_{n-1}, x_n)$, the first variable of the first atom is the input of the plan, the second variable of each atom is the first variable of its successor, and the output variables are ordered in the same way as the atoms.

\begin{Example}\label{ex1} Consider again our example in Figure~\ref{fig:ex1}. There are 3 relations names in the database: \textit{onAlbum}, \textit{sang}, and \textit{rel\-Album}. The relation \textit{rel\-Album} links a song to a related album. The functions are:
\begin{align*}
\textit{getAlbum}(\underline{s},a)  & \leftarrow onAlbum(\underline{s}, a) \\[-.15em]
\textit{getAlbumDetails}(\underline{a},s,m)  & \leftarrow onAlbum^-(\underline{a},s), sang^-(s,m) \\
\textit{getRelAlbum}(\underline{s}, a)  & \leftarrow relAlbum(\underline{s},a)
\end{align*} The first function takes as input a song $s$, and returns as output the album $a$ of the song. The second function takes as input an album $a$ and returns the songs $s$ with their musicians $m$. The last function returns the related albums of a song. 
\end{Example}

\paragraph{Execution Plans}
Our goal in this work is to study when we can evaluate an atomic query on an instance using a set of path functions, which we will do using \emph{plans}. Formally,
a \emph{plan} is a finite sequence $\pi_a(x) = c_1, \ldots, c_n$ of \emph{function calls}, where $a$ is a constant, $x$ is the output variable. Each function call $c_i$ is of the form $f(\underline{\alpha}, \beta_1, \ldots, \beta_n)$, where $f$ is a function name, where the input $\alpha$ is either a constant or a variable occurring in some call in $c_1, \ldots, c_{i-1}$, and where the outputs $\beta_1, \ldots, \beta_n$ are either variables or constants.
A \emph{filter} in a plan is the use of a constant in one of the outputs $\beta_i$ of a function call; if the plan has none, then we call it \emph{unfiltered}.
The \textit{semantics} of the plan is the query:
\[
q(x) \leftarrow \phi(c_1), \ldots, \phi(c_n)
\]
where each $\phi(c_i)$ is the body of the query defining the function $f$ of the call~$c_i$ in which we have substituted the constants and variables used in $c_i$, where we have used fresh existential variables across the different $\phi(c_i)$, and where $x$ is the output variable of the plan.

To \emph{evaluate} a plan on an instance means running the query above. Given an execution plan $\pi_a$ and a database $I$, we call $\pi_a(I)$ the answers of the plan on $I$. In practice, evaluating the plan means calling the functions in the order given by the plan. If a call fails, it can potentially remove one or all answers of the plan. More precisely, for a given instance $I$, the results $b \in \pi_a(I)$ are precisely the elements $b$ to which we can bind the output variable when matching the semantics of the plan on~$I$.
For example, let us consider a function $f(\underline{x}, y) = r(x, y)$ and a plan $\pi_a(x) = f(a, x), f(b, y)$. This plan returns the answer $a'$ on the instance $I = \{r(a, a'), r(b, b')\}$, and returns no answer on $I' = \{r(a, a')\}$.

\begin{Example}\label{ex2} The following is an execution plan for Example~\ref{ex1}:
\begin{align*}
\pi_{Jailhouse}(m) = getAlbum(\underline{Jailhouse},a), getAlbumDetails(\underline{a},\textit{Jailhouse},m)
\end{align*}
The first element is a function call to \textit{getAlbum} with the constant \textit{Jailhouse} as input, and the variable $a$ as output. The variable $a$ then serves as input in the second function call to \textit{get\-Album\-Details}. The plan is shown in Figure~\ref{fig:ex1} on page~\pageref{fig:ex1} with an example instance. This plan defines the query:
\[
onAlbum(Jailhouse,a), onAlbum^-(a, Jailhouse), sang^-(Jailhouse,m)
\]
\noindent For our example instance, we have the  embedding:
\[\sigma=\{a=JailhouseRock, m=Elvis Presley\}.
\]
\end{Example}

\paragraph{Atomic Query Rewriting}
Our goal is to determine when a given atomic query $q(x)$ can be evaluated as a plan $\pi_a(x)$.
Formally, we say that $\pi_a(x)$ is a \emph{rewriting} (or an \emph{equivalent plan}) of the query $q(x)$ if, for any database instance $I$ satisfying the inclusion dependencies $\mathcal{UID}$, the result of the plan $\pi_a$ is equal to the result of the query $q$ on~$I$.

\section{Problem Statement and Main Results}
\label{sec:problem}

The goal of this paper is to determine when a query admits a rewriting under the inclusion dependencies. If so, we compute a rewriting. In this section, we present our main high-level results for this task. We then describe in the next section (Section~\ref{sec:thealgorithm}) the algorithm that we use to achieve these results, and show in Section~\ref{sec-experiments} our experimental results on an implementation of this algorithm.

Remember that we study \textit{atomic} queries, e.g., $q(x) \leftarrow r(a, x)$, that we study plans on a set $\mathcal{F}$ of path functions, and that we assume that the data satisfy integrity constraints given as a set $\mathcal{UID}$ of \emph{unary inclusion dependencies}. In this section, we first introduce the notion of \emph{non-redundant plans}, which are a specific class of plans that we study throughout the paper; and we then state our results about finding rewritings that are non-redundant plans.

\subsection{Non-redundant plans}
\label{sec:nonredundant}

Our goal in this section is to restrict to a well-behaved subset of plans that are \emph{non-redundant}.
Intuitively, a \emph{redundant plan} is a plan that contains function calls that are not useful to get the output of the plan. For example, if we add the function call $getAlbum(m, a')$ to the plan in Example~\ref{ex2}, then this is a redundant call that does not change the result of $\pi_{Jailhouse}$. We also call \emph{redundant} the calls that are used to remove some of the answers, e.g., for the function $f(\underline{x}, y) = r(x, y)$ and the plan $\pi_a(x) = f(a, x), f(b, y)$ presented before, the second call is redundant because it does not contribute to the output (but can filter out some results).
Formally:

\begin{Definition}[Redundant plan]
An execution plan $\pi_a(x)$ is \emph{redundant} if it has no call using the constant $a$ as input, or if it contains a call where none of the outputs is an output of the plan or an input to another call. If the plan does not satisfy these conditions, it is \emph{non-redundant}. 
\end{Definition}

Non-redundant plans can easily be reformulated to have a more convenient shape: the first call uses the input value as its input, and each subsequent call uses as its input a variable that was an output of the previous call. Formally:

\begin{restatable}{Property}{corollarysequence}\label{property-call-sequence}
The function calls of any non-redundant plan $\pi_a(x)$ can be organized in a sequence $c_0, c_1,  \ldots, c_k$ such that the input of $c_0$ is the constant $a$, every other call $c_i$ takes as input an output variable of the previous call $c_{i-1}$, 
and the output of the plan is in the call $c_k$.
\end{restatable}

Non-redundant plans seem less potent than redundant plans, because
they cannot, e.g., filter the outputs of a call based on whether some other call is successful. However, as it turns out, we can restrict our study to non-redundant plans without loss of generality, which we do in the remainder of the paper.

\begin{restatable}{Property}{propertyredundantplan}\label{prop-redundant}
    For any redundant plan $\pi_a(x)$ that is a rewriting to an atomic query $q(x) \leftarrow r(a,x)$, a subset of its calls forms a non-redundant plan, which is also equivalent to $q(x)$.
\end{restatable}

\subsection{Result statements}

Our main theoretical contribution is the following theorem:

\begin{restatable}{Theorem}{thmalgorithm}\label{thm-algorithm}
There is an algorithm which, given an atomic query $q(x) \leftarrow r(a, x)$, a set $\mathcal{F}$ of path function definitions, and a set $\mathcal{UID}$ of UIDs, decides in polynomial time if there exists an equivalent rewriting of $q$. If so, the algorithm enumerates all the non-redundant plans that are equivalent rewritings of~$q$.
\end{restatable}

In other words, we can efficiently decide if equivalent rewritings exist, and when they do, the algorithm can compute them. Note that, in this case, the generation of an equivalent rewriting is \emph{not} guaranteed to be in polynomial time, as the equivalent plans are not guaranteed to be of polynomial size. Also, observe that this result gives a \emph{characterization} of the equivalent non-redundant plans, in the sense that \emph{all} such plans are of the form that our algorithm produces. Of course, as the set of equivalent non-redundant plans is generally infinite, our algorithm cannot actually write down all such plans, but it provides any such plan after a finite time. The underlying characterization of equivalent non-redundant plans is performed via a context-free grammar describing possible paths of a specific form, which we will introduce in the next section.

Our methods can also solve a different problem: given the query, path view definitions, unary inclusion dependencies, and given a candidate non-redundant plan, decide if the plan is correct, i.e., if it is an equivalent rewriting of the query. The previous result does not provide a solution as it produces all non-redundant equivalent plans in some arbitrary order. However, we can show using similar methods that this task can also be decided in polynomial time:

\begin{restatable}{Proposition}{propositionverification}\label{proposition-verification}
Given a set of unary inclusion dependencies, a set of path functions, an atomic query $q(x) \leftarrow r(a, x)$ and a non-redundant execution plan~$\pi_a$, one can determine in PTIME if $\pi_a$ is an equivalent rewriting of $q$.
\end{restatable}

That proposition concludes the statement of our main theoretical contributions. We describe in the next section the algorithm used to show our main theorem (Theorem~\ref{thm-algorithm}) and used for our experiments in Section~\ref{sec-experiments}.
The \ifthenelse{\boolean{final}}{extended version of this paper~\cite{romero2020equivalent}}{appendix} contains the proofs for our theorems.%

\section{Algorithm}
\label{sec:thealgorithm}

We now present the algorithm used to show Theorem~\ref{thm-algorithm}. The presentation explains at a high level how the algorithm can be implemented, as we did for the experiments in Section~\ref{sec-experiments}. However, some formal details of the algorithm are deferred to the \ifthenelse{\boolean{final}}{extended version of this paper~\cite{romero2020equivalent}}{appendix}, as well as the formal proof.

Our algorithm is based on a characterization of the non-redundant equivalent rewritings as the intersection between a context-free grammar and a regular expression (the result of which is itself a context-free language). The context-free grammar encodes the UID constraints and generates a language of words that intuitively describe forward-backward paths that are guaranteed to exist under the UIDs. As for the regular expression, it encodes the path functions and expresses the legal execution plans. Then, the intersection gets all non-redundant execution plans that satisfy the UIDs. We first detail the construction of the grammar, and then of the regular expression.

\subsection{Defining the context-free grammar of forward-backward paths}\label{sec:CFG}
Our context-free grammar intuitively describes a language of forward-backward paths, which intuitively describe the sequences of relations that an equivalent plan can take to walk away from the input value on an instance, and then walk back to that value, as in our example on Figure~\ref{fig:ex1}, to finally use the relation that consists of the query answer: in our example, the plan is  \emph{get\-Album(Jailhouse,$a$), get\-Album\-Details($a$,Jailhouse,$m$)}. The grammar then describes all such back-and-forth paths from the input value that are guaranteed to exist thanks to the unary inclusion dependencies that we assumed in $\mathcal{UID}$. Intuitively, it describes such paths in the \emph{chase} by $\mathcal{UID}$ of an answer fact. We now define this grammar, noting that the definition is independent of the functions in $\mathcal{F}$:

\begin{Definition}[Grammar of forward-backward paths]\label{def-grammar} Given a set of relations $\mathcal{R}$, given an atomic query $q(a, x) \leftarrow r(a,x)$ with $r \in \mathcal{R}$, and given a set of unary inclusion dependencies $\mathcal{UID}$,
the \emph{grammar of forward-backward paths} is a context-free grammar $\mathcal{G}_q$, whose language is written $\mathcal{L}_q$,
with the non-terminal symbols $S \cup \{L_{r_i},  B_{r_i}  \mid r_i \in \mathcal{R} \}$,  the terminals $\{r_i \mid r_i \in \mathcal{R}\}$, the start symbol $S$, and the following productions:

	\begin{align}
		& S\rightarrow B_r r \label{ag2} \\
		& S\rightarrow B_r r  B_{r^-} r^- \label{ag4}\\
		& \forall r_i, r_j\in \mathcal{R} \text{ s.t. } r_i \leadsto r_j \text{~in~} \mathcal{UID}: B_{r_i} 
		\rightarrow B_{r_i} L_{r_j} \label{ag5}\\
		& \forall r_i \in \mathcal{R} : B_{r_i} \rightarrow \epsilon \label{ag8}\\
		& \forall r_i\in \mathcal{R}: L_{r_i} \rightarrow r_i B_{r^-_i} r^-_i \label{ag7}
	\end{align} 
\end{Definition}

The words of this grammar describe the sequence of relations of paths starting at the input value and ending by the query relation $r$, which are guaranteed to exist thanks to the unary inclusion dependencies $\mathcal{UID}$. In this grammar, the $B_{r_i}$s represent the paths that ``loop'' to the position where they started, at which we have an outgoing $r_i$-fact. These loops are either empty (Rule~\ref{ag8}), are concatenations of loops which may involve facts implied by $\mathcal{UID}$ (Rule~\ref{ag5}), or may involve the outgoing $r_i$ fact and come back in the reverse direction using $r_i^-$ after a loop at a position with an outgoing $r_i^-$-fact (Rule~\ref{ag7}).

\subsection{Defining the regular expression of possible plans}\label{sec:RL}
While the grammar of forward-backward paths describes possible paths that are guaranteed to exist thanks to $\mathcal{UID}$, it does not reflect the set $\mathcal{F}$ of available functions. This is why we intersect it with a regular expression that we will construct from $\mathcal{F}$, to describe the possible sequences of calls that we can perform following the description of non-redundant plans given in Property~\ref{property-call-sequence}.

The intuitive definition of the regular expression is simple: we can take any sequence of relations, which is the semantics of a function in $\mathcal{F}$, and concatenate such sequences to form the sequence of relations corresponding to what the plan retrieves. However, there are several complications. First, for every call, the output variable that we use may not be the last one in the path, so performing the call intuitively corresponds to a prefix of its semantics: we work around this by adding some backward relations to focus on the right prefix when the output variable is not the last one. Second, the last call must end with the relation $r$ used in the query, and the variable that precedes the output variable of the whole plan must not be existential (otherwise, we will not be able to filter on the correct results). Third, some plans consisting of one single call must be handled separately. Last, the definition includes other technicalities that relate to our choice of so-called \emph{minimal filtering plans} in the correctness proofs that we give in the \ifthenelse{\boolean{final}}{extended version~\cite{romero2020equivalent}}{appendix}.
Here is the formal definition:

\begin{Definition}[Regular expression of possible plans]\label{possible-plan} Given a set of functions $\mathcal{F}$   and an atomic query  $q(x) \leftarrow r(a,x)$, 
for each function $f: r_1(x_0, x_1), ... r_n(x_{n-1}, x_n)$ of $\mathcal{F}$ and input or output variable $x_i$, define:
 \[
    w_{f,i}= \left\{\begin{array}{ll}
        r_1 \ldots r_i & \text{ if } i=n \\
        r_1 \ldots r_n r_n^{-} ... r^{-}_{i+1} & \text{ if }  0 \leq i < n  \\
        \end{array}\right.\\
  \]
For $f \in \mathcal{F}$ and $0 \leq i < n$, we say that a $w_{f,i}$ is \emph{final} when:
\begin{itemize}
    \item the last letter of $w_{f,i}$ is $r^-$, or it is $r$ and we have $i > 0$;
    \item writing the body of $f$ as above, the variable $x_{i+1}$ is an output variable;
    \item for $i < n-1$, if $x_{i+2}$ is an output variable, we require that $f$ does not contain the atoms: $r(x_{i}, x_{i+1}).r^-(x_{i+1}, x_{i+2})$.
\end{itemize}

\noindent The regular expression of possible plans is then $P_r=W_0 | (W^* W')$, where:
\begin{itemize}
    \item $W$ is the disjunction over all the $w_{f,i}$ above with $0 < i \leq n$.
    \item $W'$ is the disjunction over the final $w_{f,i}$ above with $0 < i < n$.
    \item $W_0$ is the disjunction over the final $w_{f,i}$ above with $i= 0$.
\end{itemize}
\end{Definition}

\subsection{Defining the algorithm}

We can now present our algorithm to decide the existence of equivalent rewritings and enumerate all non-redundant equivalent execution plans when they exist, which is what we use to show Theorem~\ref{thm-algorithm}:

\noindent \textbf{Input}: a set of path functions $\mathcal{F}$, a set of relations $\mathcal{R}$, a set of $\mathcal{UID}$ of UIDs, and an atomic query $q(x) \leftarrow r(a, x)$.\\
\noindent \textbf{Output}: a (possibly infinite) list of rewritings.
\begin{enumerate}
    \item Construct the grammar $\mathcal{G}_q$ of forward-backward paths (Definition~\ref{def-grammar}).
    \item Construct the regular expression $P_r$  of possible plans (Definition~\ref{possible-plan}).
    \item Intersect $P_r$ and $\mathcal{G}_q$ to create a grammar $\mathcal{G}$
    \item Determine if the language of $\mathcal{G}$ is empty:
        \begin{enumerate}
            \item[] If \underline{no}, then no equivalent rewritings exist and stop;
            \item[] If \underline{yes}, then continue   
        \end{enumerate}
    \item For each word $w$ in the language of $\mathcal{G}$:
    \begin{itemize}
    \item For each  
    execution plan $\pi_a(x)$ that can be built from $w$ (intuitively decomposing $w$ using $P_r$, see \ifthenelse{\boolean{final}}{extended version~\cite{romero2020equivalent}}{appendix} for details):
    \begin{itemize}
        \item For each subset $S$ of output variables of~$\pi_a(x)$:
        \begin{itemize}
    \item If adding a filter to~$a$ on the outputs in~$S$ gives an equivalent plan, then output the plan (see \ifthenelse{\boolean{final}}{extended version~\cite{romero2020equivalent}}{appendix} for how to decide this)
    \end{itemize}
     \end{itemize}\end{itemize}
\end{enumerate}

\noindent Our algorithm thus decides the existence of an equivalent rewriting by computing the intersection of a context-free language and a regular language and checking if its language is empty.
As this problem can be solved in PTIME, 
the complexity of our entire algorithm is polynomial in the size of its input.
The correctness proof of our algorithm (which establishes Theorem~\ref{thm-algorithm}), and the variant required to show Proposition~\ref{proposition-verification}, are given in the \ifthenelse{\boolean{final}}{extended version of this paper~\cite{romero2020equivalent}}{appendix}.

\section{Experiments}\label{sec-experiments}

We have given an algorithm that, given an atomic query and a set of path functions, generates all equivalent plans for the query (Section~\ref{sec:thealgorithm}). We now compare our approach experimentally to two other methods, Susie~\cite{susie}, and PDQ~\cite{benedikt2014pdq}, on both synthetic datasets and real functions from Web services.
\\[-2em]

\subsection{Setup}

We found only two systems that can be used to rewrite a query into an equivalent execution plan: Susie~\cite{susie} and PDQ (Proof-Driven Querying)~\cite{benedikt2014pdq}. We benchmark them against our implementation. All algorithms must answer the same task: given an atomic query and a set of path functions, produce an equivalent rewriting, or claim that there is no such rewriting.

We first describe the Susie approach. Susie takes as input a query and a set of Web service functions and extracts the answers to the query both from the functions and from Web documents. Its rewriting approach is rather simple, and we have reimplemented it in Python. However, the Susie approach is not complete for our task: she may fail to return an equivalent rewriting even when one exists. What is more, as Susie is not looking for equivalent plans and makes different assumptions from ours, the plan that she returns may not be equivalent rewritings (in which case there may be a different plan which is an equivalent rewriting, or no equivalent rewriting at all).

Second, we describe PDQ. The PDQ system is an approach to generating query plans over semantically interconnected data sources with diverse access interfaces. We use the official Java release of the system. PDQ runs the chase algorithm~\cite{ALICE} to create a canonical database, and, at the same time, tries to find a plan in that canonical database. If a plan exists, PDQ will eventually find it; and whenever PDQ claims that there is no equivalent plan, then indeed no equivalent plan exists. However, in some cases, the chase algorithm used by PDQ may not terminate. In this case, it is impossible to know whether the query has a rewriting or not. We use PDQ by first running the chase with a timeout, and re-running the chase multiple times in case of timeouts while increasing the search depth in the chase, up to a maximal depth. The exponential nature of PDQ's algorithm means that already very small depths (around 20) can make the method run for hours on a single query.

Our method is implemented in Python and follows the algorithm presented in the previous section. For the manipulation of formal languages, we used pyformlang\footnote{\url{https://pyformlang.readthedocs.io}}. Our implementation is available online\footnote{\url{https://github.com/Aunsiels/query_rewriting}}. All experiments were run on a laptop with Linux, 1 CPU with 4 cores at  2.5GHz, and 16 GB RAM.

\subsection{Synthetic Functions}\label{sec-synthetic}

In our first experiments, we consider a set of artificial relations $\mathcal{R}=\{r_1,...,r_n\}$, and randomly generate path functions up to length 4. Then we tried to find a equivalent plan for each query of the form $r(c,x)$ for $r\in\mathcal{R}$. The set $\mathcal{UID}$ consists of all pairs of relations $r \leadsto s$ for which there is a function in whose body $r^-$ and $s$ appear in two successive atoms. We made this choice because functions without these UIDs are useless in most cases.

For each experiment that we perform, we generate 200 random instances of the problem, run each system on these instances, and average the results of each method.
Because of the large number of runs, we had to put a time limit of 2 minutes per chase for PDQ and a maximum depth of 16 (so the maximum total time with PDQ for each query is 32 minutes). In practice, PDQ does not strictly abide by the time limit, and its running time can be twice longer.
We report, for each experiment, the following numbers:
\begin{itemize}
    \item Ours: The proportion of instances for which our approach found an equivalent plan. As our approach is proved to be correct, this is the true proportion of instances for which an equivalent plan exists.
    \item Susie: The proportion of instances for which Susie returned a plan which is actually an equivalent rewriting (we check this with our approach).
    \item PDQ: The proportion of instances for which PDQ returned an equivalent plan (without timing out): these plans are always equivalent rewritings.
    \item Susie Requires Assumption: The proportion of instances for which Susie returned a plan, but the returned plan is not an equivalent rewriting (i.e., it is only correct under the additional assumptions made by Susie).
    \item PDQ Timeout: The proportion of instances for which PDQ timed out (so we cannot conclude whether a plan exists or not).
\end{itemize}
In all cases, the two competing approaches (Susie and PDQ) cannot be better than our approach, as we always find an equivalent rewriting when one exists, whereas Susie may fail to find one (or return a non-equivalent one), and PDQ may timeout. The two other statistics (Susie Requires Assumption, and PDQ Timeout) denote cases where our competitors fail, which cannot be compared to the performance of our method.

\begin{figure}[t]
\captionsetup[subfigure]{oneside,margin={1cm,0cm}}
\subfloat[]{
    \includegraphics[width=0.5\linewidth]{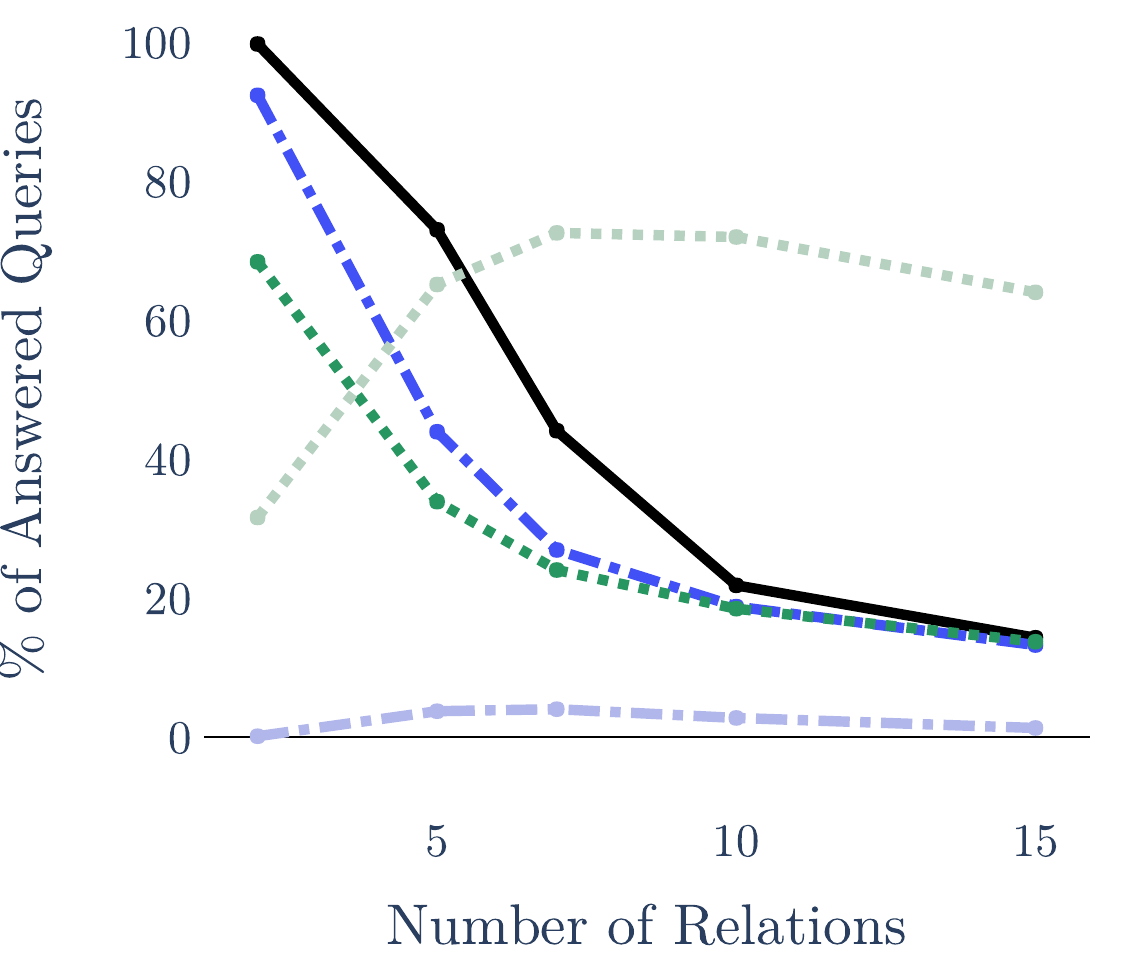}
    \label{answerednrelations}
} 
\captionsetup[subfigure]{oneside,margin={1cm,0cm}}
\hspace{-.5cm}
\raisebox{-.07cm}{
\subfloat[]{
    \includegraphics[width=0.5\linewidth]{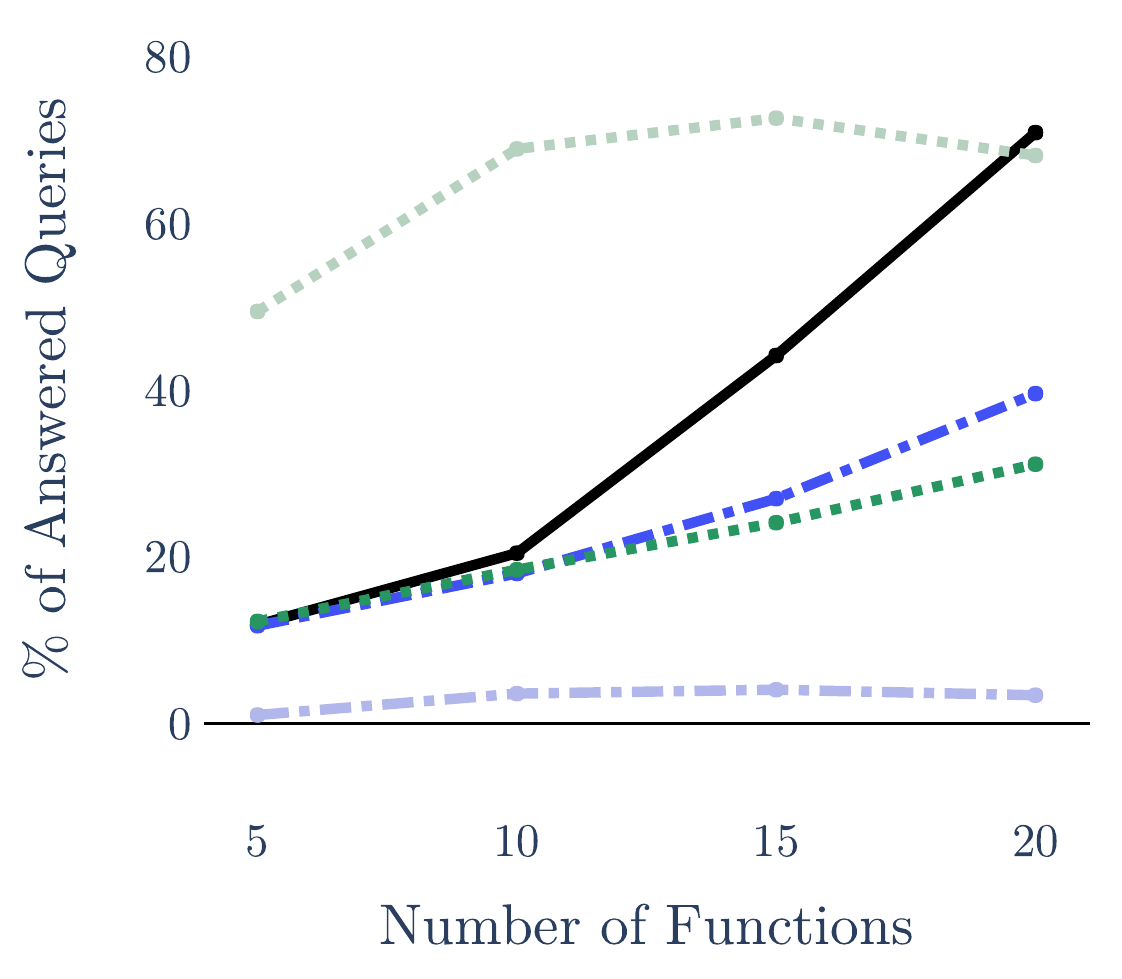}
    \label{answerednfunctions}
}
}
\\[-0.4cm]
\subfloat[]{
    \includegraphics[width=0.5\linewidth]{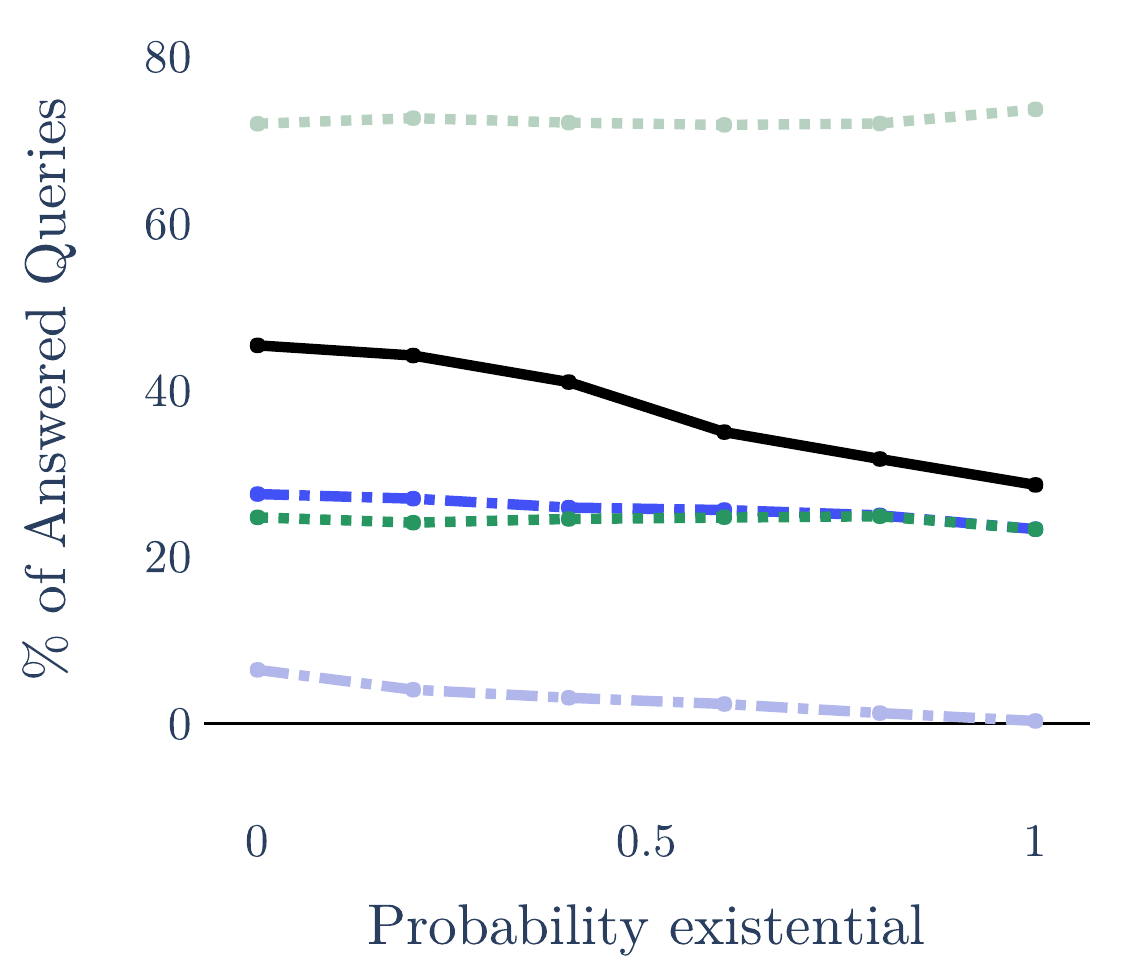}
    \label{answerednproba}
}\hfill
\subfloat[]{
    \includegraphics[width=0.37\linewidth]{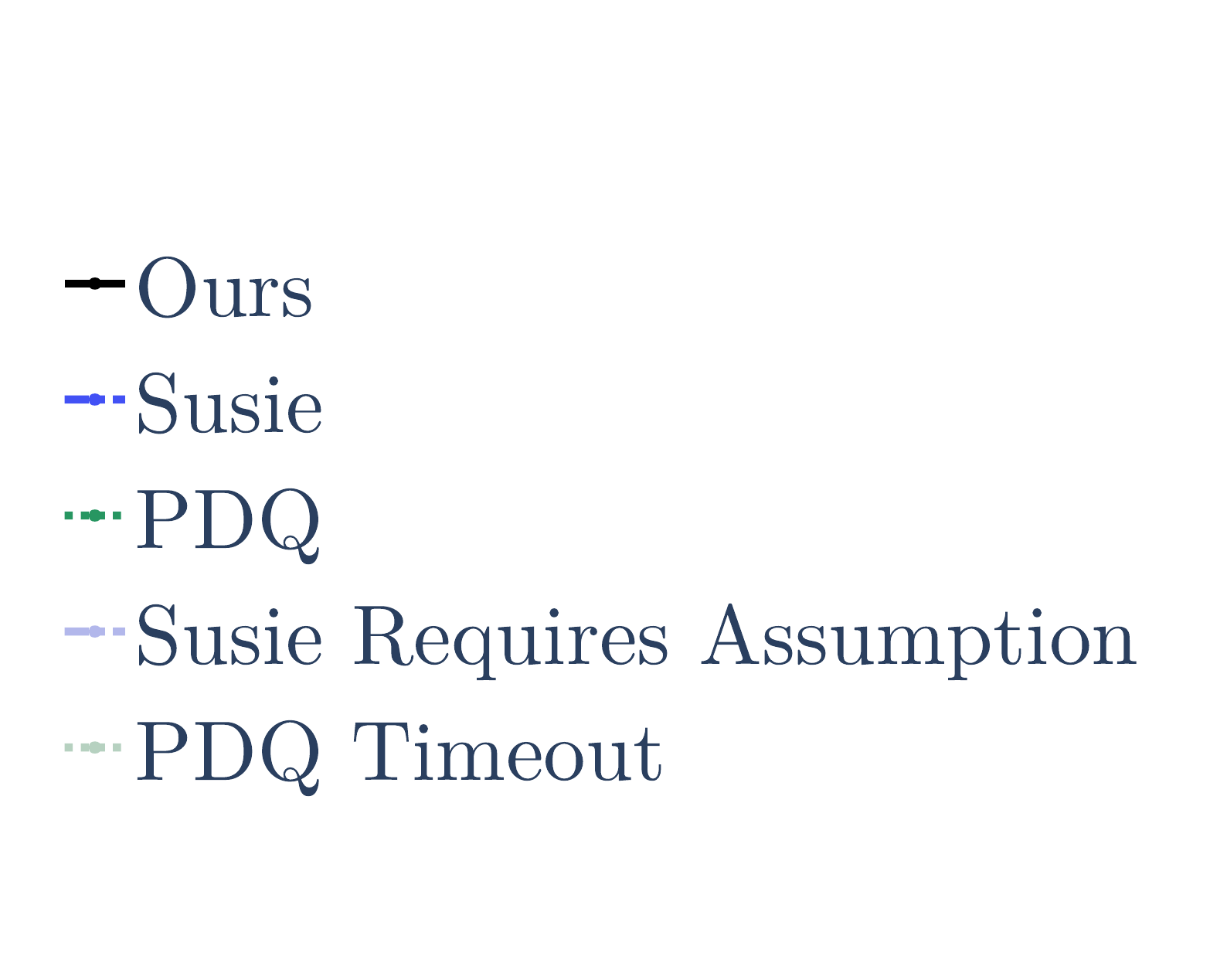}
}
\hfill\null
\caption{Percentage of answered queries  with varying number of (a) relations, (b) functions, and (c) existential variables; (d) key to the plots. %
}%
\end{figure}

In our first experiment, we limited the number of functions to 15, with 20\% of existential variables, and varied the number $n$ of relations. Both Susie and our algorithm run in less than 1 minute in each setting for each query, whereas PDQ may timeout.
Figure~\ref{answerednrelations} shows which percentage of the queries can be answered. As expected, when the number of relations increases, the rate of answered queries decreases as it becomes harder to combine functions. 
Our approach can always answer strictly more queries than Susie and PDQ.

In our next experiment, we fixed the number of relations to 7, the probability of existential variables to 20\%, and varied the number of functions. Figure~\ref{answerednfunctions} shows the results. As we increase the number of functions, we increase the number of possible function combinations. Therefore, the percentage of answered queries increases both for our approach and for our competitors. However, our approach answers about twice as many queries as Susie and PDQ.

In our last experiment, we fixed the number of relations to 7, the number of functions to 15, and we varied the probability of having an existential variable. Figure \ref{answerednproba} shows the results. As we increase the probability of existential variables, the number of possible plans decreases because fewer outputs are available to call other functions. However, the impact is not as marked as before, because we have to impose at least one output variable per function, which, for small functions, results in few existential variables. As Susie and PDQ use these short functions in general, changing the probability did not impact them too much. Still, our approach can answer about twice as many queries as Susie and PDQ.

\subsection{Real-World Web Services}

\begin{table}[t]
\centering
  \caption{Web services and results}
  {
  \renewcommand{\tabcolsep}{9pt}
  \label{real-results}%
       	\begin{tabularx}{\linewidth}{@{~~}Xrrrrrr@{~~}}
			\toprule
			Web Service & Functions & \ignore{Extracted Path Functions &} Relations & Susie & PDQ (timeout) & Ours\\
			\midrule
			Movies & 2 & \ignore{20 &} 8 & 13\% & \textbf{25}\% (0\%) & \textbf{25}\% \\
			Books & 13 & \ignore{136 &} 28 & 57\% & 64\% (7\%) & \textbf{68}\% \\
			Music & 24 & \ignore{162 &} 64 & 22\% & 22\% (25\%) & \textbf{33}\% \\
  \bottomrule 
  \end{tabularx}
  }
\end{table}

We consider the functions of Abe Books (\url{http://search2.abebooks.com}), ISBNDB (\url{http://isbndb.com/}), LibraryThing (\url{http://www.librarything.com/}), and MusicBrainz (\url{http://musicbrainz.org/}), all used in~\cite{susie}, and Movie DB (\url{https://www.themoviedb.org}) to replace the (now defunct) Internet Video Archive used in~\cite{susie}.
We add to these functions some other functions built by the Susie approach.
We group these Web services into three categories: Books, Movies, and Music, on which we run experiments separately. For each category, we manually map all services into the same schema and generate the UIDs as in Section~\ref{sec-synthetic}.
Our dataset is available online (see URL above).

\begin{table*}[t]
\centering
\caption{Examples of real functions}\label{examplesFunctions}
\begin{tabularx}{\linewidth}{X}
  \toprule
\textbf{GetCollaboratorsByID}(\underline{artistId}, collab, collabId) $\leftarrow$ \\
$~~~$hasId$^-$(artistId,artist),
isMemberOf(artist,collab), hasId(collab,collabId)
\ \\
\textbf{GetBookAuthorAndPrizeByTitle}(\underline{title}, author, prize) $\leftarrow$ \\
 $~~~$isTitled$^-$(title, book), wrote$^-$(book,author), hasWonPrize(author,prize)
\ \\
\textbf{GetMovieDirectorByTitle}(\underline{title}, director) $\leftarrow$ \\
$~~~$ isTitled$^-$(title,movie), directed$^-$(movie,director) \\
  \bottomrule
\end{tabularx}
\end{table*}

The left part of Table~\ref{real-results} shows the number of functions and the number of relations for each Web service.
Table~\ref{examplesFunctions} gives examples of functions. Some of them are recursive. For example, the first function in the table allows querying for the collaborators of an artist, which are again artists. This allows for the type of infinite plans that we discussed in the introduction, and that makes query rewriting difficult.

\begin{table*}[t]
	\caption{Example plans}\label{examples}%
	\begin{tabularx}{\linewidth}{Xl}
		\toprule
		Query \ignore{& Skeleton word} & Execution Plan\\
		\midrule
  released & \ignore{& hasArtistId, hasArtistId$^-$, $\star$, released, $\star$, hasArtistId, hasArtistId$^-$, $\star$, sang, sang$^-$, released$^-$ } GetArtistInfoByName, GetReleasesByArtistID, GetArtistInfoByName,\\
  &GetTracksByArtistID, GetTrackInfoByName, GetReleaseInfoByName \\
  published & \ignore{published, $\star$, wrote$^-$, wrote, published$^-$&} GetPublisherAuthors, GetBooksByAuthorName \\
  actedIn & \ignore{yactedIn, $\star$, isTitled, isTitled$^-$, actedIn$^-$&} GetMoviesByActorName, GetMovieInfoByName \\
		\bottomrule
	\end{tabularx}
\end{table*}

For each Web service, we considered all queries of the form $r(c,x)$ and $r^-(c,x)$, where $r$ is a relation used in a function definition. We ran the Susie algorithm, PDQ, and our algorithm for each of these queries. The runtime is always less than 1 minute for each query for our approach and Susie but can timeout for PDQ. The time limit is set to 30 minutes for each chase, and the maximum depth is set to 16. Table~\ref{real-results} shows the results, similarly to Section~\ref{sec-synthetic}. As in this case, all plans returned by Susie happened to be equivalent plans, we do not include the ``Susie Requires Assumption'' statistic (it is $0\%$).
Our approach can always answer more queries than Susie and PDQ,
and we see that with more complicated problems (like Music), PDQ tends to timeout more often.

In terms of the results that we obtain, some queries can be answered by rather short execution plans. Table~\ref{examples} shows a few examples. However, 
our results show that many queries do not have an equivalent plan. In the Music domain, for example, it is not possible to answer \emph{produced}$(c,x)$ (i.e., to know which albums a producer produced), \emph{hasChild$^-$(c,x)} (to know the parents of a person), and \emph{rated$^-$}$(c,x)$ (i.e., to know which tracks have a given rating). 
This illustrates that the services maintain control over the data, and do not allow arbitrary requests.

\section{Conclusion}

In this paper, we have addressed the problem of finding equivalent execution plans for Web service functions. We have characterized these plans for atomic queries and path functions, and we have given a correct and complete method to find them. Our experiments have demonstrated that our approach can be applied to real-world Web services and that its completeness entails that we always find plans for more queries than our competitors. All experimental data, as well as all code, is available at the URL given in Section~\ref{sec-experiments}.
We hope that our work can help Web service providers to design their functions, and users to query the services more efficiently. For future work, we aim to broaden our results to non-path functions.
We also intend to investigate connections between our theoretical results and the methods by Benedikt et al. \cite{benedikt2016generating}, in particular possible links between our techniques and those used to answer regular path queries under logical constraints \cite{bienvenu2015regular}.

\paragraph{Acknowledgements}
Partially supported by the grants ANR-16-CE23-0007-01 (``DICOS'') and ANR-18-CE23-0003-02 (``CQFD'').
\bibliographystyle{plain}
\bibliography{bib}

\newpage
\ifthenelse{\boolean{final}}{}{\appendix

\section{Proofs on Non-Redundant Plans (Section~\ref{sec:nonredundant})}

In this first section of the appendix, we give proofs pertaining to non-redundant plans (Section~\ref{sec:nonredundant}). We introduce in particular the notion of \emph{well-filtering plans} (Section~\ref{sec:wellfiltering}), which will be useful later.

The next section of the appendix (Appendix~\ref{sec-approach}) gives a high-level presentation of key technical results about \emph{minimal filtering plans} and \emph{capturing languages}. These claims are then proved in Appendix~\ref{ap:proofs}. Last, we give in Appendix~\ref{apx:main-text-proofs} the proofs of the missing details of our main claims (Section~\ref{sec:problem}) and of our algorithm (Section~\ref{sec:thealgorithm}).

\subsection{Proof of the Structure of Non-Redundant Plans (Property~\ref{property-call-sequence})}

\corollarysequence*

By definition of a non-redundant plan, there is an atom using the constant $a$ as input. Let us call this atom $c_0$. Let us then define the sequence $c_0, c_1, \ldots, c_i$, and let us assume that at some stage we are stuck, i.e., we have chosen a call $c_i$ such that none of the output variables of $c_i$ are used as input to another call. If the output of the plan is not in $c_i$, then $c_i$ witnesses that the plan is redundant. Otherwise, the output of the plan is in $c_i$. If we did not have $i=k$, then any of the calls not in $c_0, c_1, \ldots, c_i$
witness that the plan is redundant. So we have $i=k$, and we have defined the sequence $c_0, c_1, \ldots, c_k$ as required.

\subsection{Well-Filtering Plans}
\label{sec:wellfiltering}

In this subsection, we introduce \emph{well-filtering plans}, which are used both to show that we can always restrict to non-redundant plans (Property~\ref{prop-redundant}, showed in the next appendix section) and for the correctness proof of our algorithm. We then show a result (Lemma~\ref{equivalent_r_atom}) showing that we can always restrict our study to well-filtering plans.

Let us first recall the notion of the \emph{chase}~\cite{ALICE}. The chase of an instance $I$ by a set $\mathcal{UID}$ of unary inclusion dependencies (UIDs) is a (generally infinite) instance obtained by iteratively solving the violations of $\mathcal{UID}$ on~$I$ by adding new facts. In particular, if $I$ already satisfies $\mathcal{UID}$, then the chase of $I$ by $\mathcal{UID}$ is equal to $I$ itself. The chase is known to be a \emph{canonical database} in the sense that it satisfies precisely the queries that are true on all completions of $I$ to make it satisfy $\mathcal{UID}$. We omit the formal definition of the chase and refer the reader to~\cite{ALICE} for details about this construction. We note the following property, which can be achieved whenever $\mathcal{UID}$ is closed under UID implication, and when we do the so-called \emph{restricted} chase which only solves the UID violations that are not already solved:

\begin{property}
\label{ppt:uidchase}
  Let $f$ be a single fact, and let $I$ be the instance obtained by applying the chase on~$f$. Then for each element $c$ of~$I_0$, for each relation $r \in \mathcal{R}$, there is at most one fact of~$I_0$ where $c$ appears in the first position of a fact for relation~$r$.
\end{property}

Remember now that that plans can use \emph{filters}, which allow us to only consider the results of a function call where some variable is assigned to a specific constant. In this section, we show that, for any plan $\pi_a$, the only filters required are on the constant $a$. Further, we show that they can be applied to a single well-chosen atom.

\begin{Definition}[Well-Filtering Plan]
\label{def:wellfiltering}
Let $q(x) \leftarrow r(a, x)$ be an atomic query. An execution plan $\pi_a(x)$ is said to be \emph{well-filtering for $q(x)$} if all filters of the plan are on the constant $a$ used as input to the first call and the semantics of $\pi_a$ contains at least an atom $r(a, x)$ or $r^-(x, a)$, where $x$ is the output variable.
\end{Definition}

We can then show :

\begin{Lemma}\label{equivalent_r_atom}
   Given an atomic query $q(a,x) \leftarrow r(a,x)$ and a set of inclusion dependencies $\mathcal{UID}$, any equivalent rewriting of $q$ must 
   be well-filtering.
\end{Lemma}

\begin{proof}
We first prove the second part. We proceed by contradiction. 
Assume that there is a non-redundant plan $\pi_a(x)$ which is an equivalent rewriting of $q(a, x)$ and which contains a constant $b \neq a$.
By Property~\ref{property-call-sequence}, the constant $b$ is not used as the input to a call (this is only the case of $a$, in atom $c_0$), so $b$ must be used as an output filter in $\pi_a$.
Now, consider the database $I = \{r(a, a')\}$, and let $I^*$ be the result of applying the chase by $\mathcal{UID}$ to~$I$.
The result of the query $q$ on~$I^*$ is $a'$, and $I^*$ satisfies $\mathcal{UID}$ by definition, however $b$ does not appear in~$I^*$ so $\pi_a$ does not return anything on~$I^*$ (its semantics cannot have a match), a contradiction.

We now prove the first part of the lemma. We use the form of Property~\ref{property-call-sequence}.
If we separate the body atoms where $a$ is an argument from those where both arguments are variables, we can write: $q'(a,x) \leftarrow A(a, x_1, x_2, \ldots x_n), B(x_1, x_2, \ldots x_n)$ where $A(a, x_1, x_2, \ldots x_n) \leftarrow r_1(a,x_1), \ldots r_n(a,x_n)$ (if we have an atom $r_i(x, a)$ we transform it into $r^-_i(a, x)$) and $a$ does not appear as argument in any of the body atoms of  $B(x_1, x_2, \ldots x_n)$. By contradiction, assume that we have $r_i \neq r$ for all $1 \leq i \leq n$.

Let $I_0$ be the database containing the single fact $r(a,b)$ and consider the database $I^*_0$ obtained by chasing the fact $r(a,b)$ by the inclusion dependencies in $\mathcal{UID}$, creating a new value to instantiate every missing fact.
Let $I^*_1= I^*_0  \cup \{r(a_1,b_1)\} \cup \{r_i(a_1, c_i) \mid r_i(a, c_i) \in I^*_0 \wedge r_i \neq r  \} \cup \{r_i(b_1, c_i) \mid r_i(b, c_i) \in I^*_0 \wedge r_i \neq r^- \}$.   By construction, $I^*_1$ satisfies  $\mathcal{UID}$.   Also, we have that $\forall r_i \neq r$,  $r_i(a,c_i) \in  I^*_1 \Leftrightarrow r_i(a_1,c_i) \in I^*_1$.  Hence,  we have that $A(a, x_1, x_2, \ldots x_n)(I^*_1) = A(a_1, x_1, x_2, \ldots x_n)(I^*_1)$. Then, given that $B(x_1, x_2, \ldots x_n)$ does not contain $a$ nor $a_1$, we have that  $q'(a_1, x)(I^*_1) = q'(a, x)(I^*_1)$. From the  hypothesis we also have that $q'(a, x)(I^*_1)=q(a, x)(I^*_1)$ and $q'(a_1, x)(I^*_1)=q(a_1, x)(I^*_1)$. This implies that $q(a_1, x)(I^*_1)=q(a, x)(I^*_1)$. Contradiction.
\end{proof}

\subsection{Proof that we Can Restrict to Non-Redundant Plans (Property~\ref{prop-redundant})}
\label{proof-prop-redundant}

 We can now prove the property claiming that it suffices to study non-redundant plans. Recall its statement:

\propertyredundantplan*

In what follows, we write $q(a, x)$ instead of $q(x)$ to clarify the inner constant. Let $\pi_a(x)$ be an equivalent plan. From Lemma~\ref{equivalent_r_atom}, we have that its semantics contains a body atom $r(a,x)$ or $r^-(x,a)$. Hence, there is a call $c$ such that $r(a,x)$ or $r^-(x,a)$ appear in its semantics. From the definition of plans, and similarly to the proof of Property~\ref{property-call-sequence},
there is a chain of calls $c_1, c_2, \ldots c_k$ such that $c_1$ takes a constant as input, $c_k=c$, and for every two consecutive calls $c_i$ and $c_{i+1}$, with $i\in \{1, \ldots k-1\}$,  there is a variable $\alpha$ such that $\alpha$  is an output variable for $c_i$ and an output variable for $c_{i+1}$. From Lemma~\ref{equivalent_r_atom}, we have that for all the calls that take a constant as input, the constant is $a$. Hence, the input of $c_1$ is $a$. Let  $\pi'_a(x)$  be the plan consisting of the calls $c_1, c_2, \ldots c_k=c$.  Note that $c$ ensures that $r(a,x)$ or $r^-(x,a)$ appear in the  semantics of $\pi'_a(x)$.

We first notice that by construction $\pi'_a(x)$ is non-redundant.
Now, if we consider the semantics of a plan as a set of body atoms, the semantics of $\pi'_a(x)$ is contained in the semantics of $\pi_a(x)$.  Hence, we have $\forall I,  \pi_a(x)(I) \subseteq \pi'_a(x)(I)$. As $\pi_a(x)$ is equivalent to $q(x) \leftarrow r(a,x)$, $\forall I$, we have  $\pi_a(x)(I) = q(x)(I)$. As $\pi'_a(x)$ contains $r(a,x)$, $\pi'_a(x)(I) \subseteq q(x)(I)$. So, $\forall I,  q(x)(I) = \pi_a(x)(I) \subseteq \pi'_a(x)(I)  \subseteq q(x)(I)$. Hence, all the inclusions are equalities, and indeed $\pi'_a(x)$ is also equivalent to the query under $\mathcal{UID}$. This concludes the proof.

\section{Capturing Languages}
\label{sec-approach}

In this section, we give more formal details on our approach, towards a proof of Theorem~\ref{thm-algorithm} and Proposition~\ref{proposition-verification}. We will show that we can restrict ourselves to a class of execution plans called \textit{minimal filtering plans} which limit the possible filters in an execution plan. Finally, we will define the notion of \emph{capturing language} and show that the language $\mathcal{L}_q$ defined in Section~\ref{sec:thealgorithm} is capturing (Theorem~\ref{thm-main-result}); and define the notion of a language \emph{faithfully representing plans} and show that the language of the regular expression $P_r$ faithfully represents plans (Theorem~\ref{thm-faithful}).
This appendix gives a high-level overview and states the theorem; the next appendix (Appendix~\ref{ap:proofs}) contains proofs for the present appendix; and the last appendix (Appendix~\ref{apx:main-text-proofs}) contains the proofs of the claims made in Sections~\ref{sec:problem} and \ref{sec:thealgorithm}.

\subsection{Minimal Filtering Plans}

Remember the definition of well-filtering plans (Definition~\ref{def:wellfiltering}). We now
simplify even more the filters that should be applied to an equivalent plan, to limit ourselves to a single filter, by introducing \emph{minimal filtering plans}.

\begin{Definition}[Minimal Filtering Plan]
Given a well-filtering plan $\pi_a(x)$ for an atomic query $q(a,x) \leftarrow r(a,x)$, let the \textit{minimal filtering plan associated to $\pi_a(x)$} be the plan $\pi'_a(x)$ that results from removing all filters from $\pi_a(x)$ and doing the following:
\begin{itemize}
    \item We take the greatest possible call $c_i$ of the plan, and the greatest possible output variable $x_j$ of call~$c_i$, such that adding a filter on~$a$ to variable $x_j$ of call~$c_i$ yields a well-filtering plan, and define $\pi_a'(x)$ in this way.
    \item If this fails, i.e., there is no possible choice of $c_i$ and $x_j$, then we leave $\pi_a(x)$ as-is, i.e., $\pi'_a(x) = \pi_a(x)$.
    \end{itemize}
\end{Definition}

Note that, in this definition, we assume that the atoms in the semantics of each function follow the order in the definition of the path function. Also, note that the minimal filtering plan  $\pi_a'(x)$ associated to a well-filtering plan is always itself well-filtering. This fact is evident if the first case in the definition applies, and in the second case, given that $\pi_a(x)$ was itself well-filtering, the only possible situation is when the first atom of the first call of $\pi_a(x)$ was an atom of the form $r(a, x)$, with $a$ being the input element: otherwise $\pi_a(x)$ would not have been well-filtering. So, in this case, $\pi_a'(x)$ is well-filtering. Besides, note that, when the well-filtering plan $\pi_a$ is non-redundant, then this is also the case of the minimal filtering plan $\pi_a^{min}$ because the one filter that we may add is necessarily at an output position of the last call.

Finally, note that a well-filtering plan is not always equivalent to the minimal filtering plan, as removing the additional filters can add some results.  However, one can easily check if it is the case or not. This theorem is proven in Appendix~\ref{apx:prf-well-filtering-minimal-construction}.

\begin{restatable}{Theorem}{thmwellfilteringminimalconstruction}\label{thm-well-filtering-minimal-construction}
Given a query $q(x) \leftarrow r(a,x)$, a well-filtering plan $\pi_a$, the associated minimal filtering plan $\pi_a^{min}$ 
and unary inclusion dependencies:
\begin{itemize}
    \item If $\pi_a^{min}$ is not equivalent to $q$, then neither is $\pi_a$. 
    \item If $\pi_a^{min}$ is equivalent to $q$, then we can determine in polynomial time if $\pi_a$ is equivalent to $\pi_a^{min}$
\end{itemize}
\end{restatable}

This theorem implies that, when the query has a rewriting as a well-filtering plan, then the corresponding minimal filtering plan is also a rewriting:

\begin{Corollary}\label{thm-well-filtering-minimal}
Given unary inclusion dependencies, if a well-filtering plan is a rewriting for an atomic query~$q$, then it is equivalent to the associated minimal filtering plan.
\end{Corollary}

\begin{proof}
    This is the contrapositive of the first point of the theorem: if $\pi_a$ is equivalent to $q$, then so in $\pi_a^{min}$, hence $\pi_a$ and $\pi_a^{min}$ are then equivalent.
\end{proof}

For that reason, to study equivalent rewritings, we will focus our attention on minimal filtering plans: Theorem~\ref{thm-well-filtering-minimal} can identify other well-filtering plans that are rewritings, and we know by Lemma~\ref{equivalent_r_atom} that plans that are not well-filtering cannot be rewritings.

\subsection{Path Transformations}

We now show how to encode minimal filtering plans as words over an alphabet whose letters are the relation names in $\mathcal{R}$.
The key is to rewrite the plan so that its semantics is a path query of the following form:

Here is the formal notion of a \emph{path query}:

\begin{Definition}\label{skeleton}
A \textbf{path  query} is a query of the form
\[q_a(x_i)\leftarrow r_1(a,x_1), r_2(x_1,x_2), ..., r_{n}(x_{n-1},x_n)\]
where $a$ is a constant, $x_i$ is the output variable,
each $x_j$ except $x_i$ is either a variable or the constant~$a$, and $1 \leq i \leq n$.
The sequence of relations $r_1...r_n$ is called the \textbf{skeleton} of the query. 
\end{Definition}

We formalize as follows the transformation that transforms plans into path queries. We restrict it to non-redundant minimal filtering plans to limit the number of filters that we have to deal with:

\begin{Definition}[Path Transformation]
\label{def:pathtransformation}
Let $\pi_a(x)$ be a non-redundant minimal filtering execution plan and $\mathcal{R}$ a set of relations. 
We define the path transformation of $\pi_a(x)$, written $\mathcal{P}'(\pi_a)$, the transformation that maps the plan $\pi_a$ to a path query $\mathcal{P}'(\pi_a)$ obtained by applying the following steps:

\begin{enumerate}
    \item Consider the sequence of calls $c_0, c_1,..., c_k$ as defined in Property~\ref{property-call-sequence}, removing the one filter to element~$a$ if it exists.
    \item For each function call
    $c_i(\underline{y_1}, y_{i_1}..., y_{i_j},...y_{i_n}) = r_1(y_1, y_2), ..., r_k(y_k, y_{k+1}), ...,$ $r_m(y_m,$ $y_{m+1})$ in $\pi_a$ with $1 < i_1 < ... < i_n < m+1$,
    such that $y_{i_j}$ is the output used as input by the next call or is the output of the plan,
    we call the \emph{sub-semantics associated to $c_i$} the query: $r_1...r_m.r_m^-...r_{i_j}^-(y_1, ..., y_{i_j-1}, y_{i_j}',..., y_m', y_{m+1}, ..., y_{i_j})$,
    where $y_{i_j}', ..., y_m'$ are new variables.
    We do nothing if ${i_j} = m+1$.

    \item Concatenate the sub-semantics associated to the calls in the order of the sequence of calls. We call this new query the \emph{path semantics}.
    \item There are two cases:
       \begin{itemize}
           \item If the semantics of $\pi_a$ contains the atom $r(a, x)$ (either thanks to a filter to the constant~$a$ on an output variable or thanks to the first atom of the first call with $a$ being the input variable), then this atom must have been part of the semantics of the last call (in both cases). The sub-semantics of the last call is therefore of the form $\ldots, r(x_a, x'), r_2(x', x_2), \ldots, r_n(x_{n-1}, x_n), r_n^-(x_n, x_{n-1}), \ldots, r_2^-(x_2, x)$, in which $x_a$ was the variable initially filtered to $a$ (or was the input to the plan, in which case it is still the constant~$a$) and we append the atom $r^-(x, a)$ with a filter on~$a$, where $x$ is the output of the path semantics.
           \item Otherwise, the semantics of $\pi_a$ contains an atom $r^-(x, a)$, then again it must be part of the last call whose sub-semantics looks like $\ldots, r^-(x', x_2'), r_2(x_2', x_3'),...,r_n(x_{n-1}', x_n),r_n^-(x_n, x_{n-1}),..., r(x_1, x)$, in which $x_2'$ was the variable initially filtered to $a$, and we replace the last variable $x_1$ by $a$, with $x$ being the output of the path semantics.
       \end{itemize}
\end{enumerate}
\end{Definition}

We add additional atoms in the last point to ensure that the filter on the last output variable is always on the last atom of the query. Notice that the second point relates to the words introduced in Definition~\ref{possible-plan}.

The point of the above definition is that, once we have rewritten a plan to a path query, we can easily see the path query as a word in $\mathcal{R}^*$ by looking at the skeleton. Formally:

\begin{Definition}[Full Path Transformation]
\label{def:fullpathtransformation}
  Given a non-redundant minimal filtering execution plan, writing $\mathcal{R}$ for the set of relations of the signature, we denote by $\mathcal{P}(\pi_a)$ the word over $\mathcal{R}$ obtained by keeping the skeleton the path query $\mathcal{P}'(\pi_a)$ and we call it the \emph{full path transformation}.
\end{Definition}

Note that this loses information about the filters, but this is not essential.
\begin{Example}
Let us consider the two following path functions:
\begin{align*}
    f_1(x, y) & = s(x, y), t(y, z)\\
    f_2(x, y, z) &= s^-(x, y), r(y, z), u(z, z')
\end{align*}

\noindent The considered atomic query is $q(x) \leftarrow r(a, x)$. We are given the following non-redundant minimal filtering execution plan:
\[
    \pi_{a}(x) = f_1(a, y), f_2(y, a, x)
\]
\noindent We are going to apply the path transformation to $\pi_a$. Following the different steps, we have:

\begin{enumerate}
    \item The functions calls without filters are:
    \begin{align*}
        c_0(a, y) & = s(a, y), t(y, z)\\
        c_1(y, z, x) &= s^-(y, z), r(z, x), u(x, z_1)
    \end{align*}
    \item The sub-semantics associated to each function call are:
    \begin{itemize}
        \item For $c_0: s(a, y'), t(y', z), t^-(z, y)$
        \item For $c_1: s^-(y, z), r(z, x'), u(x', z_1), u^-(z_1, x)$
    \end{itemize}
    \item The path semantics obtained after the concatenation is:
    \[
    s(a, y'), t(y', z), t^-(z, y), s^-(y, z), r(z, x'), u(x', z_1), u^-(z_1, x)
    \]
    \item The semantics of $\pi_a$ contained $r(a, x)$, so add the atom $r^-(x, a)$ to the path semantics.
\end{enumerate}
\noindent
At the end of the path transformation, we get 
\[
\mathcal{P}'(\pi_a) = s(a, y'), t(y', z), t^-(z, y), s^-(y, z), r(z, x'), u(x', z_1), u^-(z_1, x), r^-(x, a)
\]
and:
\[
\mathcal{P}(\pi_a) = s, t, t^-, s^-, r, u, u^-, r^-
\]
\end{Example}

This transformation is not a bijection, meaning that possibly multiple plans can generate the same word: %

\begin{Example}
Consider three path functions:
\begin{itemize}
    \item $f_1(x, y) = s(x, y), t(y, z)$,
    \item $f_2(x, y, z) = s^-(x, y) r(y, z)$,
    \item $f_3(x, y, z) = s(x, x_0), t(x_0, x_1), t^-(x_1, x_2), s^-(x_2, x_3), r(x_3, y), r^-(y, z)$,
\end{itemize}
The execution plan $\pi_a^1(x) = f_1(a, y), f_2(y, a, x)$ then has the same image by the path transformation than the execution plan $\pi_a^2(x) = f_3(a, a, x)$.
\end{Example}

However, it is possible to efficiently reverse the path transformation whenever an inverse exists. We show this in Appendix~\ref{transducer}.

\begin{restatable}{Property}{propertytransformation}\label{prop-transformation}
Given a word $w$ in $\mathcal{R}^*$, a query $q(x) \leftarrow r(a, x)$ and a set of path functions, it is possible to know in polynomial time if there exists a non-redundant minimal filtering execution plan $\pi_a$ such that $\mathcal{P}(\pi_a) = w$. Moreover, if such a $\pi_a$ exists, we can compute one in polynomial time, and we can also enumerate all of them (there are only finitely many of them).
\end{restatable}

\subsection{Capturing Language}

The path transformation gives us a representation of a plan in $\mathcal{R}^*$. In this section, we introduce our main result to characterize \textit{minimal filtering plans}, which are atomic equivalent rewritings based on languages defined on $\mathcal{R}^*$. First, thanks to the path transformation, we introduce the notion of \emph{capturing language}, which allows us to capture equivalent rewritings using a language defined on $\mathcal{R}^*$.

\begin{Definition}[Capturing Language]
Let $q(x) \leftarrow r(a, x)$ be an atomic query. The language $\Lambda_q$ over $\mathcal{R}^*$ is said to be a \emph{capturing language} for the query $q$ (or we say that $\Lambda_q$ \emph{captures} $q$) if 
for all non-redundant minimal filtering execution plans $\pi_a(x)$, we have the following equivalence: $\pi_a$ is an equivalent rewriting of $q$ iff we have $\mathcal{P}(\pi_a) \in \Lambda_q$.
\end{Definition}

Note that the definition of capturing language does not forbid the existence of words $w \in \Lambda_q$ that are not in the preimage of~$\mathcal{P}$, i.e., words for which there does not exist a plan $\pi_a$ such that $P(\pi_a) = w$. We will later explain how to find a language that is a subset of the image of the transformation~$\mathcal{P}$, i.e., a language which \emph{faithfully represents plans}. 

Our main technical result, which is used to prove Theorem~\ref{thm-algorithm}, is that we have a context-free grammar whose language captures~$q$: specifically, the grammar $\mathcal{G}_q$ (Definition~\ref{def-grammar}):

\begin{restatable}{Theorem}{thmmainresult}\label{thm-main-result}
Given a set of unary inclusion dependencies, a set of path functions, and an atomic query $q$, the language $\mathcal{L}_q$ captures $q$.
\end{restatable}

\subsection{Faithfully representing plans}

We now move on to the second ingredient that we need for our proofs: we need a language which \emph{faithfully represents plans}:

\begin{Definition}
We say that a language $\mathcal{K}$ \emph{faithfully represents plans} (relative to a set $\mathcal{F}$ of path functions and an atomic query $q(x) \leftarrow r(a, x)$) if it is a language over $\mathcal{R}$ with the following property: for every word $w$ over $\mathcal{R}$, we have that $w$ is in~$\mathcal{K}$ iff there exists a minimal filtering non-redundant plan $\pi_a$ such that $\mathcal{P}(\pi_a) = w$.
\end{Definition}

We now show the following about the language of our regular expression $P_r$ of possible plans as defined in Definition~\ref{possible-plan}.

\begin{restatable}{Theorem}{thmfaithful}
\label{thm-faithful}
Let $\mathcal{F}$ be a set of path functions, let $q(x) \leftarrow r(a, x)$ be an atomic query, and define the regular expression $P_r$ as in Definition~\ref{possible-plan}. Then the language of~$P_r$ faithfully represents plans.
\end{restatable}

Theorems~\ref{thm-main-result} and~\ref{thm-faithful} will allow us to deduce Theorem~\ref{thm-algorithm} and Proposition~\ref{proposition-verification} from Section~\ref{sec:problem}, as explained in Appendix~\ref{apx:main-text-proofs}.

\section{Proofs for Appendix~\ref{sec-approach}}
\label{ap:proofs}

Let us first define some notions used throughout this appendix.
Recall the definition of a \emph{path query} (Definition~\ref{skeleton}) and of its skeleton.
We sometimes abbreviate the body of the query as $r_1 ... r_n(\alpha, x_1 ... x_n)$. We use the expression \emph{path query with a filter} to refer to a path query where a body variable other than $\alpha$ is replaced by a constant. For example, in Figure~\ref{fig:ex1}, we can have the path query:
\[
q(m,a) \leftarrow \text{sang}(m, s), \text{onAlbum}(s, a)
\]
which asks for the singers with their albums. Its skeleton is \emph{sang.on\-Album}.

Towards characterizing the path queries that can serve as a rewriting, it will be essential to study \emph{loop queries}:
\begin{Definition}[Loop Query]\label{suport-query}
	We call \textbf{loop query} a query of the form: $r_1 ... r_n(a,a) \leftarrow r_1(a, x_1)  ... r_n(x_{n-1},a)$  where a is a constant and $x_1, x_2, ..., x_{n-1}$ are variables such that $x_i = x_j \Leftrightarrow i=j$. 
\end{Definition}

With these definitions, we can show the results claimed in Appendix~\ref{sec-approach}.

\subsection{Proof of Theorem~\ref{thm-well-filtering-minimal-construction}}
\label{apx:prf-well-filtering-minimal-construction}

\thmwellfilteringminimalconstruction*

First, let us show the first point. By Definition~\ref{def:wellfiltering}, we have that $\pi_a$ contains $r(a,x)$ or $r^-(x, a)$. Let us suppose that $\pi_a$ is equivalent to $q$. Let $I$ be the database obtained by taking the one fact $r(a, b)$ and chasing by $\mathcal{UID}$. We know that the semantics of $\pi_a$ has a binding returning $b$ as an answer. We first argue that $\pi_a^{min}$ also returns this answer. As $\pi_a^{min}$ is formed by removing all filters from $\pi_a$ and then adding possibly a single filter, we only have to show this for the case where we have indeed added a filter. But then the added filter ensures that $\pi_a^{min}$ is well-filtering, so it creates an atom $r(a, x)$ or $r^-(x, a)$ in the semantics of $\pi_a^{min}$, so the binding of the semantics of $\pi_a$ that maps the output variable to~$b$ is also a binding of~$\pi_a^{min}$.

We then argue that $\pi_a^{min}$ does not return any other answer. In the first case, as $\pi_a^{min}$ is well-filtering, it cannot return any different answer than $b$ on~$I$. In the second case, we know by the explanation after Definition~\ref{def:wellfiltering} that $\pi_a^{min}$ is also well-filtering, so the same argument applies. Hence, $\pi_a^{min}$ is also equivalent to~$q$, which establishes the first point.

Let us now show the more challenging second point. We assume that $\pi_a^{min}$ is equivalent to $q$.
Recall the definition of a loop query (Definition~\ref{suport-query}) and the grammar $\mathcal{G}_q$ defined in Definition~\ref{def-grammar}, whose language we denoted as $\mathcal{L}_q$.
We first show the following property:

\begin{restatable}{Property}{propertyOfALoopQuery}\label{B-rule}
    A loop query $r_1 ... r_n(a,a)$ is true on all database instances satisfying the unary inclusion dependencies $\mathcal{UID}$ and containing a tuple $r(a,b)$, iff there is a derivation tree in the grammar such that $B_r \xrightarrow[]{*} r_1 ... r_n$.  
\end{restatable}

\begin{proof}
    We first show the backward direction. The proof is by structural induction on the length of the derivation. We first show the base case. If the length is 0, its derivation necessarily uses Rule~\ref{ag8}, and the query $\epsilon(a, a)$ is indeed true on all databases.
    
    We now show the induction step. Suppose we have the result for all derivations up to a length of $n - 1$. We consider a derivation of length $n > 0$. Let $I$ be a database instance satisfying the inclusion dependencies $\mathcal{UID}$ and containing the fact $r(a, b)$. Let us consider the rule at the root of the derivation tree. It can only be Rule~\ref{ag5}. Indeed, Rule~\ref{ag8} only generates words of length 0.
    So, the first rule applied was Rule~\ref{ag5} $B_r \rightarrow B_r L_{r_i}$ for a given UID $r \leadsto r_i$. Then we have two cases.
    
    The first case is when the next $B_r$ does not derive $\epsilon$ in the derivation that we study. Then, there exists $i \in \{2, \ldots, n-1\}$ such that  $B_r \xrightarrow[]{*} r_1 \ldots r_{i-1}$ and $L_{r_i} \xrightarrow[]{*} r_{i} \ldots r_n$ ($L_{r_i}$ starts by $r_i$).
    From the induction hypothesis we have that $r_1 \ldots r_{i-1}(a,a)$ has an embedding in $I$, and so has $r_i \ldots r_n(a, a)$. Indeed, we have $B_{r_i} \rightarrow L_{r_i} \xrightarrow[]{*} r_{i} \ldots r_n$ (as trivially $r_j \leadsto r_j$) and $I$ contains the tuple $r_i(a, c)$ for some constant $c$ (because we have $r \leadsto r_i$). Hence, $r_1...r_n(a, a)$ is true. This shows the first case of the induction step.
    
    We now consider the case where the next $B_{r_i}$ derives $\epsilon$. Note that, as $r(a, b) \in I$, there exists $c$ such that $r_i(a, c) \in I$.
    The next rule in the derivation is $L_{r_i} \rightarrow r_i B_{r_i^-} r_i^-$, then $B_{r_i^-} \xrightarrow[]{*} r_{2} \ldots r_{n-1}$, and $r_1 = r_i$ and $r_n=r_i^-$. By applying the induction hypothesis, we have that $r_2 \ldots r_{n-1}(c,c)$ has an embedding in $I$. Now, given that $r_1(a,c)\in I$ and $r_n(c,a) \in I$ we can  conclude that $r_1 \ldots r_{n}(a,a)$ has an embedding in $I$. This establishes the first case of the induction step. Hence, by induction, we have shown the backward direction of the proof.
    
    We now show the forward direction.
    Let $I_0$ be the database containing the single fact $r(a, b)$ and consider the database $I_0^*$ obtained by chasing the fact $r(a,b)$ by the inclusion dependencies in $\mathcal{UID}$, creating a new null to instantiate every missing fact. This database is generally infinite, and we consider a tree structure on its domain, where the root is the element $a$, the parent of $b$ is $a$, and the parent of every null $x$ is the element that occurs together with $x$ in the fact where $x$ was introduced. Now, it is a well-known fact of database theory~\cite{ALICE} that a query is true on every superinstance of $I_0$ satisfying $\mathcal{UID}$ iff that query is true on the chase $I_0^*$ of $I_0$ by $\mathcal{UID}$. Hence, let us show that all loop queries $r_1 \ldots r_n(a, a)$, which hold in $I_0^*$ are the ones that can be derived from $B_r$.
    
    We show the claim again by induction on the length of the loop query. More precisely, we want to show that, for all $n \geq 0$, for a loop query $r_1...r_n(a, a)$ which is true on all database instances satisfying the unary inclusion dependencies $\mathcal{UID}$ and containing a tuple $r(a,b)$ , we have:
    \begin{enumerate}
        \item $B_r \xrightarrow[]{*} r_1 ... r_n$
        \item For a match of the loop query on $I_0^*$, if no other variable than the first and the last are equal to $a$, then we have: $L_{r_1} \xrightarrow[]{*} r_1 ... r_n$
    \end{enumerate}
    
    If the length of the loop query is 0, then it could have been derived by the Rule~\ref{ag8}. The length of the loop query cannot be $1$ as for all relations $r'$, the query $r'(a, a)$ is not true on all databases satisfying the UIDs and containing a tuple $r(a, b)$ (for example it is not true on $I_0^*$).
    
    Let us suppose the length of the loop query is $2$ and let us write the loop query as $r_1(a, x), r_2(x, a)$ and let $r_1(a, c), r_2(c, a)$ be a match on $I_0^*$. The fact $r_1(a, c)$ can exist on $I_0^*$ iff $r \leadsto r_1$. In addition, due to the tree structure of $I_0^*$, we must have $r_2 = r_1^-$. So, we have $B_r \rightarrow L_{r_1} \rightarrow r_1 B_{r_1^-} r_1^- \rightarrow r_1^-$ and we have show the two points of the inductive claim.
    
    We now suppose that the result is correct up to a length $n - 1$ ($n > 2$), and we want to prove that it is also true for a loop query of length $n$.
    
    Consider a match $r_1(a, a_1), r_2(a_1, a_2), \ldots, r_{n-1}(a_{n-2}, a_{n-1}), r_n(a_{n-1}, a_n)$ of the loop query. Either there is some $i$ such that $a_i = a$, or there is none. If there is at least one, then let us cut the query at all positions where the value of the constant is $a$. We write the binding of the loop queries on $I_0^*$ : $(r_{i_0} \ldots r_{i_1})(a, a).(r_{i_1+1} \ldots r_{i_2})(a, a)...r_{i_{k-1} + 1} \ldots r_{i_k}(a, a)$ (where $1 = i_0 < i_1 < ... < i_{k-1} < i_k = n$). As we are on $I_0^*$, we must have, for all $0 < j < k$, that $r \leadsto r_{i_j}$. So, we can do the derivation : $B_r \rightarrow B_r L_{r_{i_{k-1}}} \rightarrow B_r L_{r_{i_{k-2}}} L_{r_{i_{k-1}}} \xrightarrow[]{*} L_{r_0} \ldots  L_{r_{i_{k-1}}}$. Then, from the induction hypothesis, we have that, for all $0 < j < k$, $L_{r_{i_j}} \xrightarrow[]{*} r_{i_j}...r_{i_{j+1}}$ and so we get the first point of our induction hypothesis.
    
    We now suppose that there is no $i$ such that $a_i = a$. Then, we still have $r \leadsto r_1$. In addition, due to the tree structure of $I_0^*$, we must have $r_n = r_1^-$ and $a_1 = a_{n-1}$. We can then apply the induction hypothesis on $r_2...r_{n-1}(a_1, a_1)$ : if it is true on all database satisfying the unary inclusion dependencies $\mathcal{UID}$ and containing a tuple $r^-(a_1,c)$, then $B_{r_1^-} \xrightarrow[]{*} r_2 ... r_{n-1}$. Finally, we observe that we have the derivation $B_r \xrightarrow[]{*} L_{r_1} \rightarrow r_1 B_{r_1^-} r_1^- \xrightarrow[]{*} r_1 ... r_n$ and so we have shown the two points of the inductive claim.
    
    Thus, we have established the forward direction by induction, and it completes the proof of the claimed equivalence.
\end{proof}

Next, to determine in polynomial time whether $\pi_a$ is equivalent to $\pi_a^{min}$ (and hence $q$),
we are going to consider all positions where a filter can be added. To do so, we need to define the \emph{root path} of a filter:

\begin{Definition}[Root Path]
Let $\pi_a$ be an execution plan. Let us consider a filter mapping a variable $y$ in the plan to a constant. Then, one can extract a unique path query $r_1...r_n(a, y)$ from the semantics of $\pi_a$, starting from the constant $a$ and ending at the variable $y$. We call this path the \emph{root path of the filter}.
\end{Definition}

The existence and uniqueness come from arguments similar to Property~\ref{property-call-sequence}: we can extract a sequence of calls to generate $y$ and then, from the semantics of this sequence of calls, we can extract the root path of the filter. Note that this is different from the definition of the path transformation (Definition~\ref{def:pathtransformation}): for each call $f(x, y_1, \ldots, y_n)$ with semantics $r_1(x, y_1), \ldots, r_n(y_{n-1}, y_n)$, if $y_i$ is the variable used in the next call or the output variable, then in the root path we only keep the path $r_1 \ldots r_i$, i.e., we do not add $r_{i+1} \ldots r_n r_n^- \ldots r_{i+1}^-$ as we did in Definition~\ref{def:pathtransformation}.

This definition allows us to characterize in which case the well-filtering plan $\pi_a$ is equivalent to its minimal filtering plan $\pi_a^{min}$, which we state and prove as the following lemma:

\begin{Lemma}
\label{lem:position-filters}
   Let $q(x) \leftarrow r(a,x)$ be an atomic query, let $\mathcal{UID}$ be a set of UIDs, and let $\pi_a^{min}$ a minimal filtering plan equivalent to $q$ under $\mathcal{UID}$. Then, for any well-filtering plan $\pi_a$ defined for $q$, the plan $\pi_a$ is equivalent to $\pi_a^{min}$ iff for each filter, letting $r_1...r_n(a, a)$ be the loop query defined from the root path of this filter,
   there is a derivation tree such that $B_r \xrightarrow[]{*} r_1 ... r_n$ in the grammar $\mathcal{G}_q$.
\end{Lemma}

It is easy to show the second point of Theorem~\ref{thm-well-filtering-minimal-construction} once we have the lemma. We have a linear number of filters, and, for each of them, we can determine in PTIME if $B_r$ generates the root path. So, the characterization can be checked in PTIME over all filters, which allows us to know if $\pi_a$ is equivalent to $\pi_a^{min}$ in PTIME, as claimed.

Hence, all that remains to do in this appendix section to establish Theorem~\ref{thm-well-filtering-minimal-construction} is to prove Lemma~\ref{lem:position-filters}. We now do so:

\begin{proof}
     We consider a filter and the root query $r_1...r_n(a, a)$ obtained from its root path.
     
     We first show the forward direction. Let us assume that $\pi_a$ is equivalent to $\pi_a^{min}$. Then, $\pi_a$ is equivalent to $q$, meaning that the loop query $r_1 ... r_n(a,a)$ is true on all database instances satisfying the unary inclusion dependencies and containing a tuple $r(a,b)$. So, thanks to Property~\ref{B-rule}, we conclude that $B_r \xrightarrow[]{*} r_1 ... r_n$.

We now show the more challenging backward direction. Assume that, for all loop queries $r_1...r_n(a, a)$ obtained from the loop path of each filter, there is a derivation tree such that $B_r \xrightarrow[]{*} r_1 ... r_n$. We must show that $\pi_a^{min}$ is equivalent to~$\pi_a$, i.e., it is also equivalent to~$q$. Now, we know that $\pi_a^{min}$ contains an atom $r(a, x)$ or~$r^-(x, a)$, so all results that it returns must be correct. All that we need to show is that it returns all the correct results. It suffices to show this on the canonical database: let $I$ be the instance obtained by chasing the fact $r(a, b)$ by the unary inclusion dependencies. As $\pi_a^{min}$ is equivalent to~$q$, we know that it returns~$b$, and we must show that $\pi_a$ also does. We will do this using the observation that all path queries have at most one binding on the canonical database, which follows from Property~\ref{ppt:uidchase}.

Let us call $\pi_a^{\text{no filter}}$ the execution plan obtained by removing all filters from $\pi_a$.
As we have $B_r \xrightarrow[]{*} r_1 ... r_n$ for all root paths, we know from Property~\ref{B-rule} that $r_1...r_n(a, a)$ is true on all databases satisfying the UIDs, and in particular on $I$. In addition, on $I$, $r_1...r_n(a, x_1,... x_n)$ has only one binding, which is the same than $r_1...r_n(a, x_1,... x_{n-1}, a)$. So, the filters of~$\pi_a$ do not remove any result of~$\pi_a$ on $I$ relative to~$\pi_a^{\text{no filter}}$: as the reverse inclusion is obvious, we conclude
that $\pi_a$ is equivalent to $\pi_a^{\text{no filter}}$ on $I$.

Now, if $\pi_a^{min}$ contains no filter or contains a filter which was in $\pi_a$, we can apply the same reasoning and we get that $\pi_a^{min}$ is equivalent to $\pi_a^{\text{no filter}}$ on $I$, and so $\pi_a^{min}$ and $\pi_a$ are equivalent in general.

The only remaining case is when $\pi_a^{min}$ contains a filter which is not in~$\pi_a$. In this case,
we have that the semantics of $\pi_a$ contains two consecutive atoms $r(a, x)r^-(x, y)$ where one could have filtered on $y$ with $a$ (this is what is done in $\pi_a^{min}$). Let us consider the root path of~$\pi$ to $y$. It is of the form $r_1...r_n(a, a) r(a, x) r^-(x, y)$. We have $B_r \xrightarrow[]{*} r_1 ... r_n$ by hypothesis. In addition, as $r \leadsto r$ trivially, we get $B_r \rightarrow B_r L_r \rightarrow B_r r r^- \xrightarrow[]{*} r_1...r_n.r.r^-$. So, $r_1...r_n.r.r^-(a,a)$ is true on $I$ (Property~\ref{B-rule}). Using the same reasoning as before, $\pi_a^{min}$ is equivalent to $\pi_a^{\text{no filter}}$ on $I$, and so $\pi_a^{min}$ and $\pi_a$ are equivalent in general.
This concludes the proof.

\end{proof}

\subsection{Proof of Property~\ref{prop-transformation}}
\label{transducer}

We show that we can effectively reverse the path transformation, which will be crucial to our algorithm:

\propertytransformation*

We are going to construct a finite-state transducer that can reverse the path transformation and give us a sequence of calls. To find one witnessing plan, it will suffice to take one run of this transducer and take the corresponding plan, adding a specific filter which we know is correct. If we want all witnessing plans, we can simply take all possible outputs of the transducer.

To construct the transducer, we are going to use the regular expression $P_r$ from Definition~\ref{possible-plan}. We know that $P_r$ faithfully represents plans (Theorem~\ref{thm-faithful}), and it is a regular expression. So we will be able to build an automaton from $P_r$ on which we are going to add outputs to represent the plans.

The start node of our transducer is $S$, and the final node is $F$. The input alphabet of our transducer is $\mathcal{R}$, the set of relations. The output alphabet is composed of function names $f$ for $f \in \mathcal{F}$, the set of path functions, and of output symbols $OUT_i$, which represents the used output of a given function. We explain later how to transform an output word into a non-redundant minimal filtering plan.

First, we begin by creating chains of letters from the $w_{f,i}$ defined in Definition~\ref{possible-plan}. For a word $w_{f,i} = r_1...r_k$ (which includes the reverse atoms added at the end when $0 \leq i < n$), this chain reads the word $r_1...r_k$ and outputs nothing.

Next, we construct $W_0$ between two nodes representing the beginning and the end of $W_0$: $S_{W_0}$ and $F_{W_0}$. From $S_{W_0}$ we can go to the start of the chain of a final $w_{f, 0}$ by reading an epsilon symbol and by outputting the function name $f$. Then, at the end of the chain of a final $w_{f, 0}$, we go to $F_{W_0}$ by reading an epsilon symbol and by outputting a $OUT_1$ letter.

Similarly, we construct $W'$ between two nodes representing the beginning and the end of $W'$: $S_{W'}$ and $F_{W'}$. From $S_{W'}$ we can go to the beginning of the chain of a final $w_{f, i}$ with $0 < i < n$ (as explained in Definition~\ref{possible-plan}) by reading an epsilon symbol and by outputting the function name $f$. Then, at the end of the chain of a final $w_{f, i}$, we go to $F_{W'}$ by reading an epsilon symbol. The output symbol of the last transition depends on the last letter of $w_{f, i}$: if it is $r$, then we output $OUT_i$; otherwise, we output $OUT_{i+1}$. This difference appears because we want to create a last atom $r(a, x)$ or $r^-(x, a)$, and so our choice of output variable depends on which relation symbol we have.

Last, using the same method again, we construct $W$ between two nodes representing the beginning and the end of $W$: $S_{W}$ and $F_{W}$. From $S_{W}$ we can go to the beginning of the chain of a $w_{f, i}$ with $0 < i \leq n$ (as explained in Definition~\ref{possible-plan}) by reading an epsilon symbol and by outputting the function name $f$. Then, at the end of the chain of a final $w_{f, i}$, we go to $F_{W'}$ by reading an epsilon symbol and outputting $OUT_i$. In this situation, there is no ambiguity on where the output variable is.

Finally, we can link everything together with epsilon transitions that output nothing. We construct $W^*$ thanks to epsilon transitions between $S_{W}$ and $F_{W}$. Then, $W^*W'$ is obtained by linking $F_{W}$ to $S_{W'}$ with an epsilon transition. We can now construct $P_r = W_0|(W^*W')$ by adding an epsilon transition between $S$ and $S_{W_0}$, $S$ and $S_W$, $F_{W_0}$ and $F$ and $F_{W'}$ and $F$.

We obtain a transducer that we call $\mathcal{T}_{reverse}$.

Let $w$ be a word of $\mathcal{R}^*$. To know if there is a non-redundant minimal filtering execution plan $\pi_a$ such that $\mathcal{P}(\pi_a) = w$, one must give $w$ as input to $\mathcal{T}_{reverse}$. If there is no output, there is no such plan $\pi_a$. Otherwise, $\mathcal{T}_{reverse}$ nondeterministically outputs some words composed of an alternation of function symbols $f$ and output symbols $OUT_i$. From this representation, one can easily reconstruct the execution plan: The function calls are the $f$ from the output word and the previous $OUT$ symbol gives their input. If there is no previous $OUT$ symbol (i.e., for the first function call), the input is $a$.
If the previous $OUT$ symbol is $OUT_k$, then the input is the $k^{th}$ output of the previous function. The last $OUT$ symbol gives us the output of the plan. We finally add a filter with $a$ on the constructed plan to get an atom $r(a, x)$ or $r^-(x, a)$ in its semantics in the last possible atom, to obtain a minimal filtering plan. Note that this transformation is related to the one given in the proof of Theorem~\ref{thm-faithful} in Section~\ref{proof-thm-faithful}, where it is presented in a more detailed way.

Using the same procedure, one can enumerate all possible output words for a given input and then obtain all non-redundant minimal filtering execution plans $\pi_a$ such that $\mathcal{P}(\pi_a) = w$. We can understand this from the proof of Theorem~\ref{thm-faithful} in Section~\ref{proof-thm-faithful}, which shows that there is a direct match between the representation of $w$ as words of $w_{f,i}$ and the function calls in the corresponding execution plan. 
Last, the reason why the set of output words is finite is because the transducer must at least read an input symbol to generate each output symbol.

\subsection{Proof of Theorem~\ref{thm-main-result}}

In this appendix, we finally show the main theorem of Appendix~\ref{sec-approach}:

\thmmainresult*

Recall that $\mathcal{L}_q$ is the language of the context-free grammar $\mathcal{G}_q$ from Definition~\ref{def-grammar}. Our goal is to show that it is a capturing language.

In what follows, we say that two queries are \emph{equivalent} under a set of UIDs if they have the same results on all databases satisfying the UIDs.

\subsubsection{Linking $\mathcal{L}_q$ to equivalent rewritings.}

In this part, we are going to work at the level of the words of $\mathcal{R}^*$ ending by a $r$ or $r^-$ (in the case $q(x) \leftarrow r(a, x)$ is the considered query), where $\mathcal{R}$ is the set of relations. Recall that the full path transformation (Definition~\ref{def:fullpathtransformation}) transforms an execution plan into a word of $\mathcal{R}^*$ ending by an atom $r$ or $r^-$. Our idea is first to define which words of $\mathcal{R}^*$ are interesting and should be considered. In the next part, we are going to work at the level of functions.

For now, we start by defining what we consider to be the ``canonical'' path query associated to a skeleton. Indeed, from a skeleton in $\mathcal{R}^*$ where $\mathcal{R}$ is the set of relations, it is not clear what is the associated path query (Definition~\ref{skeleton}) as there might be filters. So, we define:

\begin{Definition}[Minimal Filtering Path Query]\label{def-loop-query}
 Given an atomic query $q(a,x) \leftarrow r(a,x)$, a set of relations $\mathcal{R}$ and a word $w \in \mathcal{R}^*$ of relation names from~$\mathcal{R}$ ending by $r$ or $r^-$, the \textbf{minimal filtering path query} of $w$ for $q$ is the path query of skeleton $w$ taking as input $a$ and having a filter such that its last atom is either $r(a,x)$ or $r^-(x,a)$, where $x$ is the only output variable.
\end{Definition}

As an example, consider the word  \emph{on\-Album.on\-Album$^-$.sang$^-$}. The minimal filtering path query is: \emph{$q'($Jailhouse, $x) \leftarrow$ on\-Album$($Jailhouse,$y)$, on\-Album$^-(y,$ Jailhouse$)$, sang$^-($Jailhouse$, x)$}, which is an equivalent rewriting of the atomic query  \emph{sang$^-($Jailhouse$, x)$}.

We can link the language $\mathcal{L}_q$ of our context-free grammar to the equivalent rewritings by introducing a corollary of Property~\ref{B-rule}:

\begin{restatable}{Corollary}{corollaryForwardBackwardQuery} \label{lemma-fwd-backward} Given an atomic query $q(a,x) \leftarrow r(a,x)$ and a set $\mathcal{UID}$ of UIDs, the minimal filtering path query of any word in $\mathcal{L}_q$
is a equivalent to $q$. Reciprocally, for any query equivalent to $q$ that could be the minimal filtering path query of a path query of skeleton $w$ ending by $r$ or $r^-$, we have that $w \in \mathcal{L}_q$. 
\end{restatable}

Notice that the minimal filtering path query of a word in $\mathcal{L}_q$ is well defined as all the words in this language end by $r$ or $r^-$.

\begin{proof}
    We first suppose that we have a word $w \in \mathcal{L}_q$. We want to show that the minimal filtering path query of $w$ is equivalent to $q$. We remark that the minimal filtering path query contains the atom $r(a,x)$ or $r^-(x, a)$. Hence, the answers of the given query always include the answers of the minimal filtering path query, and we only need to show the converse direction.
    
    Let $I$ be a database instance satisfying the inclusion dependencies $\mathcal{UID}$ and let $r(a,b) \in I$ (we suppose such an atom exists, otherwise the result is vacuous). Let $q'(a,x)$ be the head atom of the minimal filtering path query. It is sufficient to show that $q'(a,b)$ is an answer of the minimal filtering path query to prove the equivalence.  We proceed by structural induction. Let $w \in \mathcal{L}_q$.  Let us consider a bottom-up construction of the word.
    The last rule can only be one of the Rule~\ref{ag2} or the Rule~\ref{ag4}.
    If it is Rule~\ref{ag2}, then $\exists r_1, \ldots, r_n \in \mathcal{R}$ such that $w=r_1 \ldots r_n r$ and $B_r \xrightarrow[]{*} r_1 \ldots r_n$. By applying Property~\ref{B-rule}, we know that $r_1 \ldots r_n(a,a)$ has an embedding in $I$. Hence, $q'(a,b)$ is an answer.
    If the rule is Rule \ref{ag4}, then
    $\exists r_1, \ldots, r_n, \ldots, r_m\in \mathcal{R}$
    such that $w= r_1 \ldots r_n r r_{n+1} \ldots r_m r^-$,
    $B_{r} \xrightarrow[]{*} r_1 \ldots r_n$ and $B_{r^-} \xrightarrow[]{*} r_{n+1} \ldots r_m$. By applying Property~\ref{B-rule} for the two derivations, and remembering that we have $r(a, b)$ and $r^-(b, a)$ in $I$, we have that $r_1 \ldots r_n(a,a)$ and $r_{n+1} \ldots r_m(b,b)$ have an embedding in $I$.  Hence, also in this case, $q'(a,b)$ is an answer. We conclude that $q'$ is equivalent to $q$.
    
    Reciprocally, let us suppose that we have a minimal filtering path query of a path query of skeleton $w$, which is equivalent to $q$, and that $q'(a, x)$ is its head atom.
    We can write it either $q'(a, x) \leftarrow r_1(a, x_1), r_2(x_1, x_2),... , r_n(x_{n-1}, a) r(a, x)$ or $q'(a, x) \leftarrow r_1(a, x_1), r_2(x_1, x_2),... ,$ $r_n(x_{n-1}, x), r^-(x, a)$. 
    In the first case, as $q'$ is equivalent to $q$, we have $r_1...r_n(a, a)$ which is true on all databases $I$ such that $I$ contains a tuple $r(a, b)$. So, according to Property~\ref{B-rule}, $B_r \xrightarrow[]{*} r_1 ... r_n$, and using Rule~\ref{ag2}, we conclude that $r_1...r_n.r$ is in $\mathcal{L}_q$. In the second case, for similar reasons, we have $B_r \xrightarrow[]{*} r_1 ... r_n r^-$. The last $r^-$ was generated by Rule~\ref{ag7}, using a non-terminal $L_r$ which came from Rule~\ref{ag5} using the trivial UID $r \leadsto r$. So we have $B_r \rightarrow B_r L_r \rightarrow B_r r B_{r^-} r^- \xrightarrow[]{*} r_1 ... r_n r^-$. We recognize here Rule~\ref{ag4} and so $r_1 ... r_n r^- \in \mathcal{L}_q$. This shows the second direction of the equivalence, and concludes the proof.
\end{proof}

\subsubsection{Linking the path transformation to $\mathcal{L}_q$.}
\label{proof-thm-capturing}

In the previous part, we have shown how equivalent queries relate to the context-free grammar $\mathcal{G}_q$ in the case of minimal filtering path queries. We are now going to show how the full path transformation relates to the language $\mathcal{L}_q$ of $\mathcal{G}_q$, and more precisely, we will observe that the path transformation leads to a minimal filtering path query of a word in $\mathcal{L}_q$.

The path transformation operates at the level of the semantics for each function call, transforming the original tree-shaped semantics into a word. %
What we want to show is that after the path transformation, we obtain a minimal filtering path query equivalent to $q$ iff the original execution path was an equivalent rewriting.
To show this, we begin by a lemma:

\begin{Lemma}\label{last_is_r}
   Let $\pi_a$ a minimal filtering non-redundant execution plan. The query $\mathcal{P}'(\pi_a)(x)$ is a minimal filtering path query and it ends either by $r(a, x)$ or $r^-(x, a)$, where $x$ was the variable name of the output of $\pi_a$.
\end{Lemma}

\begin{proof}
    By construction, $\mathcal{P}'(\pi_a)(x)$ is a minimal filtering path query. Let us consider its last atom.
    In the case where the original filter on the constant $a$ created an atom $r(a, x)$, then the result is clear: an atom $r^-(x, a)$ is added.
    Otherwise, it means the original filter created an atom $r^-(x, a)$. Therefore, as observed in the last point of the path transformation, the last atom is $r(a, x)$, where we created the new filter on $a$.
\end{proof}

The property about the preservation of the equivalence is expressed more formally by the following property:

\begin{restatable}{Property}{propertyEquivalencePathTransformation}\label{prop:equivalence-path-transformation}
    Let us consider a query $q(x) \leftarrow r(a, x)$, a set of unary inclusion dependencies $\mathcal{UID}$, a set of path functions $\mathcal{F}$ and a minimal filtering non-redundant execution plan $\pi_a$ constructed on $\mathcal{F}$.
    Then, $\pi_a$ is equivalent to $q$ iff the minimal filtering path query $\mathcal{P}'(\pi_a)$ is equivalent to $q$.
\end{restatable}

\begin{proof}
    First, we notice that we have $\pi_a(I) \subseteq q(I)$ and $\mathcal{P}'(\pi_a)(I) \subseteq q(I)$ as $r(a, x)$ or $r^-(a, x)$ appear in the semantics of $\pi_a(x)$ and in $\mathcal{P}'(\pi_a)(x)$ (see Lemma~\ref{last_is_r}). So, it is sufficient to prove the property on the canonical database $I_0$ obtained by chasing the single fact $r(a, b)$ with $\mathcal{UID}$.
    
    We first show the forward direction and suppose that $\pi_a$ is equivalent to $q$. Then, its semantics has a single binding on $I_0$. Let us consider the $i^{th}$ call $c_i(\underline{x_1}, ..., x_j,...x_n)$ ($x_j$ is the output used as input by another function or the output of the plan) in $\pi_a$ and its binding: $r_1(y_1, y_2), ..., r_k(y_k, y_{k+1}), ..., r_m(y_m, y_{m+1})$. Then $r_1(y_1, y_2),$ $...,$ $r_{k-1}(y_{k-1}, y_k), r_k(y_k, y_{k+1}),$ $...,$ $r_m(y_m, y_{m+1}), r_m^-(y_{m+1}, y_m),$ $...,$ $r_k^-(y_{k+1},$ $y_{k})$ is a valid binding for the sub-semantics, as the reversed atoms can be matched to the same atoms than those used to match the corresponding forward atoms.
    Notice that the last variable is unchanged. So, in particular, the variable named $x$ (the output of $\pi_a$) has at least a binding in~$I_0$ before at step 3 of the path transformation.
    In step 4, we have two cases.
    In the first case, we add $r^-(x, a)$ to the path semantics. Then we still have the same binding for $x$ as $\pi_a$ is equivalent to $q$.
    In the second case, we added a filter in the path semantics, and we still get the same binding.
    Indeed, by Property~\ref{ppt:uidchase}, $b$ has a single ingoing $r$-fact in $I_0$, which is $r^-(b, a)$. This observation means that, in the binding of the path semantics, the penultimate variable was necessary $a$ on $I_0$.
    We conclude that $\mathcal{P}'(\pi_a)$ is also equivalent to $q$.

    We now show the backward direction and suppose $\mathcal{P}'(\pi_a)$ is equivalent to $q$. Let us take the single binding of $\mathcal{P}'(\pi_a)$ on $I_0$. Let us consider the sub-semantics of a function call $c_i(\underline{x_1}, ..., x_j,...x_n)$ (where $x_j$ is the output used as input by another function or the output of the plan): $r_1(y_1, y_2),$ $...,$ $r_{k-1}(y_{k-1}, y_k'), r_k(y_k', y_{k+1}'),$ $...,$ $r_m(y_m', y_{m+1}), r_m^-(y_{m+1}, y_m),$ $...,$ $r_k^-(y_{k+1},$ $y_{k})$. As all path queries have at most one binding on $I_0$, we necessarily have $y_k = y_k'$, ..., $y_m = y_m'$. Thus the semantics of $\pi_a$ has a binding on $I$ which uses the same values than the binding of $\mathcal{P}'(\pi_a)$. In particular, the output variables have the same binding.
    We conclude that $\pi_a$ is also equivalent to $q$.
\end{proof}

We can now apply Corollary~\ref{lemma-fwd-backward} on $\mathcal{P}'(\pi_a)$ as it is a minimal filtering path query : $\mathcal{P}'(\pi_a)$ is equivalent to $q$ iff $\mathcal{P}(\pi_a)$ is in $\mathcal{L}_q$. 
So, $\pi_a$ is equivalent to $q$ iff $\mathcal{P}(\pi_a)$ is in $\mathcal{L}_q$.

We conclude that $\mathcal{L}_q$ is a capturing language for $q$.
As our grammar $\mathcal{G}_q$ for~$\mathcal{L}_q$ can be constructed in PTIME, this concludes the proof of Theorem~\ref{thm-main-result}. 

\subsection{Proof of Theorem~\ref{thm-faithful}}
\label{proof-thm-faithful}

\thmfaithful*

\begin{proof}
    We first prove the forward direction: every word $w$ of the language of~$P_r$ is achieved as the image by the full path transformation of a minimal filtering non-redundant plan.
    To show this, let $w$ be a word in the language of $P_r$.
    We first suppose we can decompose it into its elements from $W$ and $W'$: $w = w_{f_1, i_1} ... w_{f_{n-1}, i_{n-1}}.w_{f_n, i_n}$ with $w_{f_n, i_n}$ being final.
    Let $\pi_a$ be composed of the successive function calls $f_1, \ldots, f_n$ where the input of $f_1$ is the constant $a$ and the input of $f_k$ ($k > 1$) is the $i_{k-1}^{th}$ variable of $f_{k-1}$. For the output variable and the filter, we have two cases:
    \begin{enumerate}
        \item  If the last letter of $w_{f_n, i_n}$ is $r$, the output of the plan is the $i_n^{th}$ variable of $f_n$ and we add a filter to~$a$ on the $(i_n+1)^{th}$ variable.
        \item Otherwise, if the last letter of $w_{f_n, i_n}$ is $r^-$, the output of the plan is the $(i_n+1)^{th}$ variable of $f_n$ and we add a filter to~$a$ on the $i_n^{th}$ variable (except if this is the input, in which case we do nothing).
    \end{enumerate}
    We notice that $\pi_a$ is non-redundant. Indeed, by construction, only the first function takes $a$ as input, and all functions have an output used as input in another function. The added filter cannot be on the input of a function as $i_n > 0$.
    What is more, $\pi_a$ is also a minimal filtering plan. Indeed, by construction, we create an atom $r(a, x)$ or an atom $r^-(x, a)$ (with $x$ the output of the plan). Let us show that it is the last possible filter. If we created $r^-(x, a)$, it is obvious as $x$ cannot be used after that atom. If we created $r(a, x)$, we know we could not have a following atom $r^-(x, y)$ where one could have filtered on $y$: this is what is guaranteed by the third point of the definition of the final $w_{f,i}$.

    The only remaining point is to show that $\mathcal{P}(\pi_a) = w$. Indeed, for $k < n$, we notice that $w_{f_k, i_k}$ is the skeleton of the sub-semantics of the $k^{th}$ call in $\pi_a$. What is less intuitive is what happens for the last function call.
    
    Let us consider the two cases above. In the first one, the output variable is the $i_n^{th}$ variable of $f_n$. We call it $x$. The semantics of $\pi_a(x)$ contains $r_{i_n+1}(x, a) = r^-(x, a)$ (as the last letter of $w_{f_n, i_n}$ is $r$). We are in the second point of step 4 of the path transformation. The skeleton of the end of the path semantics is not modified and it is $w_{f_n, i_n}$.%
    
    In the second case, the output variable is the $(i_n + 1)^{th}$ variable of $f_n$. We call it $x$. The semantics of $\pi_a(x)$ contains $r_{i_n+1}(a, x) = r(a, x)$ (as the last letter of $w_{f_n, i_n}$ is $r^-$). We are in the first point of step 4 of the path transformation. The skeleton of the end of the path semantics is modified to append $r^-$ and it is now $w_{f_n, i_n + 1} r = w_{f_n, i_n}$ as expected.

    This establishes that $\mathcal{P}(\pi_a) = w$ in the case where $w$ can be decomposed as elements of $W$ and $W'$.
    Otherwise, $w$ is in the language of $W_0$, so $w=w_{f,0}$ where $w_{f,0}$ is final, 
    ends by $r^-$, thus starts by $r$. We define $\pi_a$ as the execution plan composed of one function call $f$, which takes as input $a$. The output of the plan is the first output variable of the function. The plan $\pi_a$ is non-redundant as it contains only one function call. It is also minimal filtering. Indeed, by definition of $w_{f,0}$, the first output variable is on the first atom.
    So, the semantics of $\pi_a$ contains an atom $r(a, x)$ where $x$ is the output variable. Besides, it does not contain an atom $r^-(x, y)$ where $y$ is an output of the $f$ by the third point of the definition of a final $w_{f,i}$.
    
    Finally, we have $\mathcal{P}(\pi_a) = w$, and this concludes the first implication. The transformation that we have here is what was performed by the transducer and the method presented in Section~\ref{transducer}. The only difference is that the technique in Section~\ref{transducer} will consider all possible ways to decompose $w$ into $w_{f, i}$ and into final $w_{f, i}$ to get all possible non-redundant minimal filtering plans.

    We now show the converse direction of the claim: the full path transformation maps non-redundant minimal filtering plans to words in the language of~$P_r$.
    Suppose that we have a non-redundant minimal filtering plan $\pi_a$ such that $\mathcal{P}(\pi_a) = w$ and let us show that $w$ is in the language of~$P_r$.
    For all calls which are not the last one, it is clear that the sub-semantics of these calls are the $w_{f, i_k}$ with $i_k > 0$ (as the plan is non-redundant).
    So the words generated by the calls that are not the last call are words of the language of~$W^*$.
    
    For the last function call, we have several cases to consider.
    
    First, if $\pi_a(x)$ contains a filter, then it means that either the first atom in the semantics of $\pi_a(x)$ is not $r(a, x)$ or, if it is, it is followed by an atom $r^-(x, a)$. 
    
    If we are in the situation where the semantics of $\pi_a(x)$ starts by $r(a, x), r^-(x, a)$, then $\pi_a$ is composed of only one function call $f$ (otherwise, it would be redundant). Then, it is clear that $\mathcal{P}(\pi_a) = w_{f, 1}$ with $w_{f,1}$ being final, and we have the correct form.
    
    If we are in the situation where the semantics of $\pi_a(x)$ does not start by $r(a, x)$, we have the two cases (corresponding to the two cases of the forward transformation). We suppose that the last function call is on $f$, and the output variable is the $i^{th}$ one in $f$.
    
    If $\pi_a$ does not contain an atom $r(a, x)$, then it contains an atom $r^-(x, a)$ and the result is clear: the skeleton of the path semantics is not modified and ends by the sub-semantics of $f$ whose skeleton is $w_{f, i}$ and has the correct properties: the last atom is $r$, the variable after $x$ is not existential (it is used to filter) and the atom after $r(a, x)$ cannot be $r^-(x, y)$ with $y$ an output variable of $f$ as $\pi_a$ is minimal filtering.
    
    If $\pi_a$ contains an atom $r(a, x)$, then in the definition of the path transformation, we append an atom $r^-(x, a)$ after the sub-semantics of the $f$. We then have the path semantics ending by the atom names $w_{f, i}.r^- = w_{f, i-1}$ and $w_{f, i-1}$, which is final, has the adequate properties.
    
    This shows that $w$ is in the language of~$P_r$ in the case $\pi_a$ has a filter because the word generated by the last call is in~$W'$.
    
    Now, we consider the case when $\pi_a$ does not have a filter. It means that the semantics of $\pi_a$ starts by $r(a, x)$ and is not followed by an atom $r^-(x, y)$ where $y$ is the output of a function (as $\pi_a$ is well-filtering). Then, $\pi_a$ is composed of only one function call $f$ and it is clear that $\mathcal{P}(\pi_a) = w_{f, 0}$ which is final. So, in this case, the word $w$ belongs to~$W_0$.
    
    So, $w$ is in the language of~$P_r$ in the case $\pi_a$ does not have a filter. This concludes the proof of the second direction, which establishes the property.
\end{proof}

\section{Proofs for Section~\ref{sec:problem} and~\ref{sec:thealgorithm}}
\label{apx:main-text-proofs}

In this section, we give the missing details for the proof of the claims given in the main text, using the results from the previous appendices. We first cover 
in Appendix~\ref{apx-algodetails}
the missing details of our algorithm. We then show Theorem~\ref{thm-algorithm} in Appendix~\ref{apx:algorithmproof}, and show Proposition~\ref{proposition-verification} in Appendix~\ref{apx:verificationproof}.

\subsection{Details for our algorithm}
\label{apx-algodetails}

We now make more precise the last steps of our algorithm, which were left unspecified in the main text:

\begin{itemize}
    \item Building all possible execution plans $\pi_a(x)$ from a word $w$ of $\mathcal{G}$: this is specifically done by taking all preimages of $w$ by the path transformation, which is done as shown in  Property~\ref{prop-transformation}. Note that these are all minimal filtering plans by definition.
    \item Checking subsets of variables on which to add filters: for each minimal filtering plan, we remove its filter, and then consider all possible subsets of output variables where a filter could be added, so as to obtain a well-filtering plan which is equivalent to the minimal filtering plan that we started with. (As we started with a minimal filtering plan, we know that at least some subset of output variables will give a well-filtering plan, namely, the subset of size 0 of 1 that had filters in the original minimal filtering plan.) The correctness of this step is because we know by Lemma~\ref{equivalent_r_atom} that non-redundant equivalent plans must be well-filtering, and because we can determine using Theorem~\ref{thm-well-filtering-minimal-construction} if adding filters to a set of output variables yields a plan which is still an equivalent rewriting. 
\end{itemize}

\subsection{Proof of Theorem~\ref{thm-algorithm}}
\label{apx:algorithmproof}

In this appendix, we show our main theorem:

\thmalgorithm*

We start by taking the grammar $\mathcal{G}_q$ with language $\mathcal{L}_q$ used in Theorem~\ref{thm-main-result} and defined in Definition~\ref{def-grammar} and the regular expression $P_r$ used in Theorem~\ref{thm-faithful} and defined in Definition~\ref{possible-plan}. We make the following easy claim:

\begin{Property}
\label{prop-intersection}
$\mathcal{L}_q \cap P_r$ is a capturing language that faithfully represents plans, and it can be constructed in PTIME.
\end{Property}

\begin{proof}
    By construction, $P_r$ represents all possible skeletons obtained after a full path transformation (Theorem~\ref{thm-faithful}).
    
    So, as $P_r$ represents all possible execution plans, and as $\mathcal{L}_q$ is a capturing language (proof of Theorem~\ref{thm-main-result}), then  $\mathcal{L}_q \cap P_r$ is a capturing language.
    
    The only remaining part is to justify that it can be constructed in PTIME. First, observe that the grammar $\mathcal{G}_q$ for $\mathcal{L}_q$, and the regular expression for~$P_r$, can be computed in PTIME. Now, to argue that we can construct in PTIME a context-free grammar representing their intersection, we will use the results of \cite{IndexGrammarHopcroft} (in particular, Theorem~7.27 of the second edition). First, we need to convert the context-free grammar $\mathcal{G}_q$ to a push-down automaton accepting by final state, which can be done in PTIME. Then, we turn $P_r$ into a non-deterministic automaton, which is also done is PTIME. Then, we compute a push-down automaton whose language is the intersection between the push-down automaton and the non-deterministic automaton using the method presented in~\cite{IndexGrammarHopcroft}. This method is very similar to the one for intersecting two non-deterministic automata, namely, by building their product automaton. This procedure is done in PTIME. In the end, we obtain a push-down automaton that we need to convert back into a context-free grammar, which can also be done in PTIME. So, in the end, the context-free grammar $\mathcal{G}$ denoting the intersection of $\mathcal{L}_q$ and of the language of~$P_r$ can be constructed in PTIME. This concludes the proof.
\end{proof}

So let us now turn back to our algorithm and show the claims. By Property~\ref{prop-intersection}, we can construct a grammar for the language $\mathcal{L}_q \cap P_r$ in PTIME, and we can then check in PTIME if the language of this new context-free grammar is empty. If it is the case, we know that is no equivalent plan. Otherwise, we know there is at least one.
We can thus generate a word $w$ of the language of the intersection -- note that this word is not necessarily of polynomial-size, so we do not claim that this step runs in PTIME. Now, as $P_r$ faithfully represents plans (Theorem~\ref{thm-faithful}), we deduce that there exists an execution plan $\pi_a$ such that $\mathcal{P}(\pi_a) = w$, and from Property~\ref{prop-transformation}, we know we can inverse the path transformation in PTIME to gut such a plan.

To get all plans, we enumerate all words of $\mathcal{L}_q \cap P_r$: each of them has at least one equivalent plan in the preimage of the full path transformation, and we know that the path transformation maps every plan to only one word, so we never enumerate any duplicate plans when doing this. Now, by Property~\ref{prop-transformation}, for any word $w \in \mathcal{L}_q \cap P_r$, we can list all its preimages by the full path transformation; and for any such preimage, we can add all possible filters, which is justified by Theorem~\ref{thm-well-filtering-minimal-construction} and Property~\ref{B-rule}.
That last observation establishes that our algorithm indeed produces precisely the set of non-redundant plans that are equivalent to the input query under the input unary inclusion dependencies, which allows us to conclude the proof of Theorem~\ref{thm-algorithm}.

\subsection{Proof of Proposition~\ref{proposition-verification}}
\label{apx:verificationproof}

\propositionverification*

First, we check if $\pi_a$ is well-filtering, which can easily be done in PTIME. If not, using Lemma~\ref{equivalent_r_atom} we can conclude that $\pi_a$ is not an equivalent rewriting. Otherwise, we check if $\pi_a$ is equivalent to its associated minimal filtering plan. This verification is done in PTIME, thanks to Theorem~\ref{thm-well-filtering-minimal-construction}. If not, we know from Theorem~\ref{thm-well-filtering-minimal} that $\pi_a$ is not an equivalent rewriting. Otherwise, it is sufficient to show that $\pi_a^{min}$ is an equivalent rewriting. To do so, we compute $w = \mathcal{P}(\pi_a^{min})$ in PTIME and check if $w$ is a word of the context-free capturing language defined in Theorem~\ref{thm-main-result}. This verification is done in PTIME. By Theorem~\ref{thm-main-result}, we know that $w$ is a word of the language iff $\pi_a$ is an equivalent rewriting, which concludes the proof.

}

\end{document}